\documentclass[a4paper,12pt]{article}
\pdfoutput=1 
\usepackage{geometry}
\usepackage[margin=10mm]{caption}
\usepackage[utf8]{inputenc}
\usepackage[english]{babel}
\usepackage{amssymb, amsmath, amsthm, amsfonts, eucal, mathalpha, mathrsfs, yfonts, lmodern, slashed, microtype}
\usepackage[sorting=none, backend=biber, natbib=true, style=numeric-comp]{biblatex}
\usepackage{array, booktabs}
\usepackage{graphicx, subcaption, tikz, tikz-cd}
\usetikzlibrary{angles, quotes, calc, shapes.geometric}
\usepackage{authblk, csquotes, hyperref}
\hypersetup{colorlinks = false , citecolor=orange, linkcolor=orange, linkbordercolor = orange, citebordercolor = orange, urlbordercolor = orange}

\theoremstyle{definition}
\newtheorem{definition}{Definition}[section]
\newtheorem{example}{Example}[section]
\newtheorem{theorem}{Theorem}[section]
\newtheorem{prop}{Proposition}[section]
\newtheorem{lemma}{Lemma}[section]

\newtheorem{remark}{Remark}[section]

\newcommand{\beq}{\begin{equation}}
\newcommand{\eeq}{\end{equation}}
\newcommand{\bqa} {\begin{eqnarray}}
\newcommand{\eqa} {\end{eqnarray}}

\newcommand{\p}{\partial}
\newcommand{\ups}{\upsilon}
\newcommand{\eps}{\varepsilon}

\def \l {\left(}
\def \r {\right)}
\def \lal {\langle}
\def \ral {\rangle}

\DeclareMathOperator{\Ad}{Ad}

\DeclareMathOperator{\Tr}{Tr}

\newcommand{\CO}{{\mathcal O}}

\newcommand{\CV}{{\mathcal V}}

\newcommand{\NN}{{\mathbb N}}
\newcommand{\ZZ}{{\mathbb Z}}
\newcommand{\RR}{{\mathbb R}}
\newcommand{\CCC}{{\mathbb C}}

\newcommand{\cstar}[1]{{\mathcal{#1}}}

\newcommand{\hilb}[1]{{\mathcal #1}}

\newcommand{\InvPhases}{{\mathcal{I}_{2d}}}
\newcommand{\InvPhasesZN}{{\mathcal{I}^{\ZZ/N}_{2d}}}

\title{Reflection positivity and a refined index for 2d invertible phases}
\author{Nikita Sopenko$\footnote{School of Natural Sciences, Institute for Advanced Study, 1 Einstein Drive, Princeton, NJ 08540 USA}$}
\date{\today}

\begin{document}
\maketitle

\begin{abstract}
We analyze the validity of reflection positivity in the classification of invertible phases of quantum spin systems. We provide a mathematical model in which every 2d invertible state admits a reflection-positive representative. We prove that reflection positivity provides a canonical lift from the set of invertible phases to the set of invertible phases protected by a $\ZZ/N$-rotational symmetry. Using this, we define a refined version of the index recently introduced by the author. This refined version conjecturally provides a microscopic characterization of an invariant that coincides with the chiral central charge $c_-$ when conformal field theory effectively describes the boundary modes.
\end{abstract}

\vspace{.7cm}

\section{Introduction}

Topological phases of quantum many-body systems correspond to equivalence classes of short-range correlated ground states at zero temperature. Two ground states represent the same phase if there is an adiabatic locality-preserving unitary evolution that transforms one state into another, possibly after addition of unentangled degrees of freedom.

It is conjectured that at least a large class of $d$-dimensional topological phases of quantum many-body systems can be classified by field theoretic methods \cite{kapustin2014symmetry, freed2016reflection}. In this approach, one assumes that macroscopic behavior can be effectively described by a $(d+1)$-dimensional unitary quantum field theory with relativistic symmetry. Positivity of energy gives a way to analytically continue such theories to Euclidean $D=d+1$-dimensional space where unitarity manifests itself in the condition called {\it reflection positivity} \cite{osterwalder1973axioms}. It states that for any codimension one hyperplane $\Sigma$ in $\RR^D$ and any observable $\CO$ supported in a half-space on one side of $\Sigma$, correlation functions satisfy\footnote{In the case of systems with fermionic degrees of freedom, a modified inequality holds.} $\lal \CO \overline{\CO}' \ral \geq 0$, where $\overline{\CO}'$ is an observable obtained by complex conjugation and reflection in $\Sigma$ (see Fig. \ref{fig:OsterwalderSchroder}). In the Hamiltonian picture of the original theory in Minkowski signature, it implies that vector spaces associated with constant time slices are equipped with a positive definite inner product. Further, it implies that for any codimension one hyperplane in $\RR^d$ and any element $x$ of the quantum algebra of observables $\cstar{A}$ that is supported on one side of this hyperplane, we have $\omega(x j(x)) \geq 0$, where $j$ is an anti-linear automorphism of $\cstar{A}$ implementing the $CRT$-transformation with respect to the hyperplane, and $\omega: \cstar{A} \to \CCC$ is a positive linear functional on $\cstar{A}$ that corresponds to the ground state. For these reasons, it is sometimes said that unitary relativistic field theories have an extended version of unitarity or reflection positivity.

From the perspective of a microscopic quantum many-body system, the assumptions of relativistic symmetry and extended reflection positivity are not justified. A microscopic model explicitly breaks spatial symmetries, therefore states can not be invariant under spatial deformations on the nose. Even if invariance under spatial isometries somehow emerges at large distances, it is unclear why one should expect an emergence of relativistic invariance. With that, it is unclear why we should expect an extended version of reflection positivity. Of course, since we start with a quantum mechanical system from the beginning, we do have unitarity in time. But a ground state $\omega$ of a generic quantum system would not satisfy the condition $\omega(x j(x)) \geq 0$. At best, one can hope that only some states in a given topological phase are reflection positive. 

This discussion raises natural questions. Is it true that states related by spatial deformations are in the same phase? Is it true that a given class of topological phases has reflection positive representatives? 

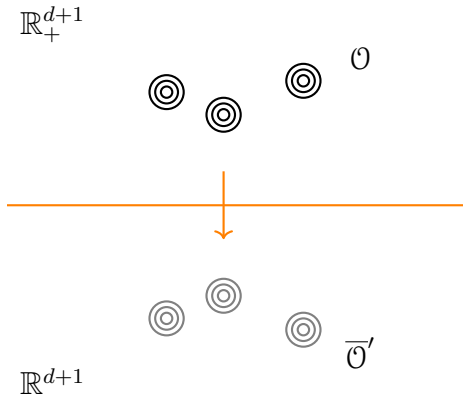
\begin{figure}
    \centering
\begin{tikzpicture}[scale=1.5]

\coordinate (A) at (-2, 2, 0);
\coordinate (B) at (2, 2, 0);
\coordinate (C) at (2, -2, 0);
\coordinate (D) at (-2, -2, 0);


\draw[thick,color=orange] (-2, 0, 0) -- (2, 0, 0);
\node at (-1.6, 1.6, 0) {$\mathbb{R}_+^{d+1}$};
\node at (-1.6, -1.6, 0) {$\mathbb{R}_-^{d+1}$};

\foreach \x/\y in {-0.6/1.0, -0.1/0.8, 0.6/1.1} {
  \foreach \r in {0.05, 0.1, 0.15} {
    \draw[thick] (\x,\y,0) circle (\r);
  }
}
\node at (1.1,1.3,0) {$\CO$};

\foreach \x/\y in {-0.6/-1.0, -0.1/-0.8, 0.6/-1.1} {
  \foreach \r in {0.05, 0.1, 0.15} {
    \draw[thick,gray] (\x,\y,0) circle (\r);
  }
}
\node at (1.1,-1.3,0) {$\overline{\CO}'$};

\draw[->, thick, color=orange] (-0.1,0.3,0) -- (-0.1,-0.3,0) node[midway,left] {};

\end{tikzpicture}
\caption{An observable $\CO$ and its time reversed reflection $\overline{\CO}'$.}
\label{fig:OsterwalderSchroder}
\end{figure}

\subsection{Canonical purification and locality}

Given a state on the algebra of quantum observables $\cstar{A}$ that models a $d$-dimensional quantum many-body system, there is a natural way to construct a reflection positive state that coincides with the original state in a half-space. Such a procedure is called a canonical purification and is well-known in quantum mechanics. In fact, once we fix the $CRT$-action that permutes two half-spaces, a desired reflection positive state is essentially unique and is reconstructed via the canonical purification. Since the phase of a given system is supposed to be determined by a large enough piece of the material, one may expect that by starting with a short-range entangled state, restricting it to a half-space and purifying it, one gets a state in the same phase which is reflection positive. What one needs to check though is that the resulting state still has short-range entanglement.

Topological phases of one-dimensional systems are well-understood by now \cite{ogata2019classification, kapustin2020classification}. They are basically characterized by a finite amount of entanglement between any two half-lines of an infinite system on $\RR$, which is mathematically formalized in terms of the {\it split property} \cite{Matsui}. This class of states is closed under the procedure of the canonical purification, because if we start with a split state, restrict it to a half-line and canonically purify it, the resulting state would still satisfy the split property.

In two dimensions, the situation is already more complicated. To see the nontriviality of the question, let us consider the following example \cite{KitaevBravyiChain}. Let $\hilb{H} = (\CCC^2)^{\otimes 2N}$, $N \in \NN$ be the Hilbert space of a periodic spin chain of length $2 N$ with the corresponding algebra of observables $\cstar{A} = M_2(\CCC)^{\otimes 2N}$. Let $X,Y,Z$ be the Pauli matrices, and let $X_i, Y_i, Z_i \in \cstar{A}$ be Pauli operators acting on site $i$, e.g., $X_i = 1 \otimes 1 \otimes ... \otimes X \otimes ... \otimes 1$, where $X$ appears at position $i$. Let $v \in \hilb{H}$ be a unit vector satisfying $X_i v =v$, $i=1,...,2N$ and $w = U v$, where
$$
U =  \l \prod_{j=1}^{2N-1} e^{\frac{\pi i}{4} Z_j Z_{j+1}} \r e^{\frac{\pi i}{4} Z_{2N} Z_{1}}.
$$
Note that $U$ is a product of strictly local commuting unitary observables. The state corresponding to $w$ is the ground state of the Hamiltonian 
$$H = Z_{2N} X_1 Z_2 + \sum_{i=2}^{2N-1} Z_{i-1} X_{i} Z_{i+1} + Z_{2N-1} X_{2N} Z_{1}$$
and is known as a {\it cluster state}. The restriction of this state to odd spins is described by the density matrix $\rho = 2^{-N}(1+X^{\otimes N})$. If we canonically purify the state in such a way that $i$-th spin corresponds to $(2N-i)$-th spin, then we get a pure state described by a vector $\tilde{w} =\frac{1}{\sqrt{2}} (1 + X_2 X_4...X_{2N}) \xi^{\otimes N}$, where $\xi \in \CCC^2 \otimes \CCC^2$ is the unit vector of an EPR pair: $(Z\otimes Z) \xi = (X \otimes X) \xi = \xi$. The resulting state $\psi$ has long-range correlations: $\psi(Z_1 Z_{2N} Z_{N} Z_{N+1}) = 1$, while $\psi(Z_1 Z_{2N}) = \psi(Z_{N} Z_{N+1}) = 0$.

For a given decomposition of $\RR^2$ into two half-spaces, we can always embed this spin chain into $\RR^2$ such that even and odd sites lie in different half-spaces (see Fig. \ref{fig:BravyiChain}) and are permuted by the corresponding $CRT$-transformation\footnote{A similar embedding appears as a counterexample to topological entanglement entropy formula and is often referred to as the Bravyi's counterexample \cite{BravyiCounterexample}}. Since this embedded chain can be prepared by a finite-time local Hamiltonian evolution and can be made arbitrarily large, it is unreasonable to expect that if we start with an arbitrary short-range entangled state with exponentially decaying correlations in the bulk, then after canonical purification we get a state with the control over its correlation length.

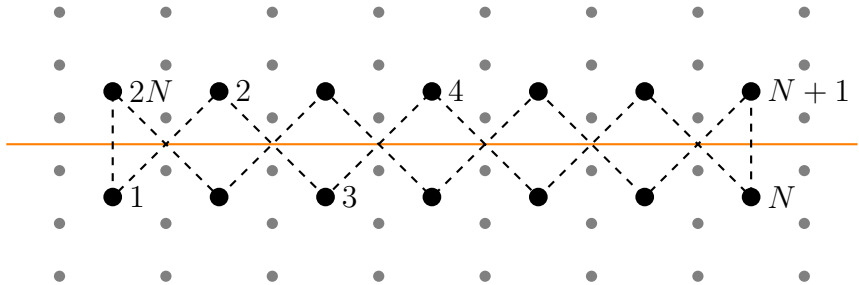
\begin{figure}
    \centering
\begin{tikzpicture}[scale = 0.7]

\draw[orange, thick] (-2, 0) -- (14, 0);

\coordinate (P1) at (0,-1);
\coordinate (P2) at (0,1);   
\coordinate (P3) at (2,-1);
\coordinate (P4) at (2,1);   
\coordinate (P5) at (4,-1);
\coordinate (P6) at (4,1);   
\coordinate (P7) at (6,-1);
\coordinate (P8) at (6,1);   
\coordinate (P9) at (8,-1);
\coordinate (P10) at (8,1);  
\coordinate (P11) at (10,-1);
\coordinate (P12) at (10,1); 
\coordinate (P13) at (12,-1);
\coordinate (P14) at (12,1); 

\draw[dashed, thick]
  (P1) -- (P2) -- (P3) -- (P6) -- (P7) -- (P10) -- (P11) --
  (P14) -- (P13) -- (P12) -- (P9) -- (P8) -- (P5) -- (P4) -- (P1);

\foreach \i in {1,...,14} {
    \fill (P\i) circle (5pt);
}

\node[right=2pt] at (P1) {$1$};
\node[right=2pt] at (P4) {$2$};
\node[right=2pt] at (P5) {$3$};
\node[right=2pt] at (P8) {$4$};
\node[right=2pt] at (P13) {$N$};
\node[right=2pt] at (P14) {$N+1$};
\node[right=2pt] at (P2) {$2N$};

\foreach \x in {-1,1,...,13} {
    \foreach \y in {-2.5,-1.5,...,2.5} {
        \fill[color=gray] ( \x, \y ) circle[radius=3pt];
    }
}

\end{tikzpicture}
    \caption{An embedding of the Bravyi chain into a spin system on $\RR^2$ separated by an orange line into two half-planes.}
    \label{fig:BravyiChain}
\end{figure}

Nevertheless, as we show in this paper, the purified state would still have a finite amount of entanglement in the thermodynamic limit between any two half-lines. More precisely, if we start with a pure state of a one-dimensional system that satisfies the split property, restrict it to any subsystem and then canonically purify it, then we would get a state that still satisfies the split property. Thus, even though we don't have a control over the correlations length, we can still characterize large-scale entanglement structure of purified states. The generalization of canonical purification to infinite-dimensional $C^*$-algebras by Woronowicz \cite{woronowicz1972purification, woronowicz1973purification} plays an essential role in the proof.

In this paper, we describe a mathematical model for invertible phases on $\RR^2$ based on states directly in the thermodynamic limit in which invertible states and states obtained by the purification of their restriction to a half-space are in the same invertible phase. While correlations in a purified state can be strong on an arbitrary large but finite scale, the state on the strip near the separating line satisfies the split property that guarantees finite entanglement between the half-strips and eventual decay of the correlations. Moreover, this model is flexible enough to show other desirable properties. In particular, the phase of an invertible state is invariant under well-behaved orientation-preserving homeomorphisms of $\RR^2$ such as rotations or scaling transformations, and states obtained by orientation-reversing homeomorphisms or by anti-unitary on-site automorphisms (time-reversal transformations) are in the opposite phase.

\subsection{Rotations, twist defects and reflection positivity}
\begin{figure}
\centering
\begin{tikzpicture}[scale=.5]

\filldraw[orange, ultra thick, fill = orange!5] (-5,-2) -- (-5+4,0+4-2) -- (5+4,0+4-2) -- (5,0-2);
\filldraw[orange, ultra thick, fill = orange!5] (-5,-1) -- (-5+4,0+4-1) -- (5+4,0+4-1) -- (5,0-1);
\filldraw[orange, ultra thick, fill = orange!5] (-5,0) -- (-5+4,0+4) -- (5+4,0+4) -- (5,0);

\draw[orange, ultra thick] (-5,-1) .. controls (-1,-1) .. (0,-1/2) .. controls (1,0) .. (5,0);
\filldraw[white, ultra thick] (-5,-2) .. controls (-1,-2) .. (0,-3/2) -- (0,-2);
\filldraw[white, ultra thick] (0,-2) -- (0,-1) .. controls (1,-2) .. (5,-2);
\filldraw[orange!5, ultra thick] (0.4, -1.2) -- (0,-1) -- (0,-3/2);
\draw[orange, ultra thick] (-5,-2) .. controls (-1,-2) .. (0,-3/2) .. controls (1,-1) .. (5,-1);
\draw[orange, ultra thick] (-5,0) .. controls (-1,0) .. (0,-1) .. controls (1,-2) .. (5,-2);

\filldraw [orange] (2,-1+2) circle (5pt);
\draw[rotate around={-45:(0,-1)}, orange, dashed] (1/2,-1) arc(0:180: 14pt and 85pt);

\end{tikzpicture}
\caption{A twist defect for $N=3$.}
\label{fig:twistdefect}
\end{figure}
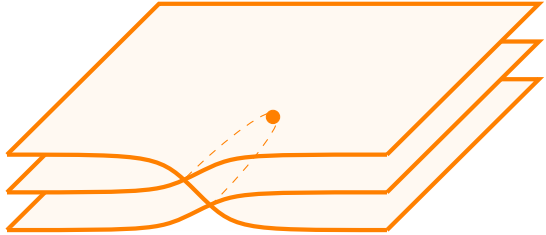

Reflection positivity is deeply tied with invariance under isometries of a Euclidean space. First of all, reflection positivity with respect to a given hyperplane automatically implies invariance under the corresponding $CRT$-symmetry. Further, if we have reflection positivity with respect to two different hyperplanes, then we automatically have an invariance under an (orientation-preserving) isometry, that is the composition of $CRT$-transformations corresponding to these hyperplanes. 

Accidental symmetries often provide us with constraints on the systems and potentially can be used to distinguish different phases. As a manifestation of this idea in this paper, we show that reflection positivity can be used to construct an index of 2d invertible phases that refines the index defined in \cite{sopenko2024}. It is conjectured that there exists a $\ZZ$-valued invariant for 2d invertible phases that under special circumstances of the emergence of conformal symmetry on the boundary is related to the chiral central charge $c_- \in 8 \ZZ$ of the Virasoro algebra \cite{kitaev2006anyons, kim2022chiral}. While it is still not proven, a step in this direction has been done in \cite{sopenko2024} where an index for invertible states has been defined. It was argued that it captures $c_- \bmod 24$. 

Let us recall the basic idea of the construction in \cite{sopenko2024}. Suppose we are given an invertible state $\psi$. By taking $N$ copies of this state, cutting them along a half-line and regluing in a way shown in Fig. \ref{fig:twistdefect}, we can construct a state that models a twist defect for an $N$-layered material. The fact that the state is invertible allows to show that such regluing can be done, and that the resulting state can be made invariant under $\ZZ/N$-group action that cyclically permutes the copies $\{1,2,...,N\} \to \{N,1,...,N-1\}$. Moreover, this regluing is unambiguous far away from the defect. It was argued that in the exchange process of a pair of such defects one gets a nontrivial universal phase factor $\theta_N$ that depends only on the $\ZZ/N$-charge of these defects. Using this idea, it was shown that $\omega_N = (\theta_N)^N$ defines an invariant of a topological phase of the original state $\psi$.

In \cite{sopenko2024}, no canonical choice of a $\ZZ/N$-charge for twist defect has been provided. However, it was argued, that if such a choice exists, then it would allow to define a refined invariant that would be given by the phase factor $\theta_N$ for twist defects with this canonical charge. In this paper, we show that such a canonical choice indeed exists. With an argument from conformal field theory, we conjecture that the resulting invariant $\theta_N$ is fine enough to capture $c_-$ as a real number, not only as $c_- \bmod 24$.

A similar ambiguity appears for systems protected by $\ZZ/N \subset SO(2)$-rotational symmetry, to which twist defect states are related by simply unwinding them onto the plane (see Fig. \ref{fig:TwistDefect2RotationallyInvariantState}). Given such an invertible state on an infinite plane, unless it is in a trivial phase, it is unclear which angular momentum is canonical. We show that if this rotationally invariant state is reflection positive, then its angular momentum is fixed. 

To explain the basic idea why it is true, let us consider a one-dimensional periodic spin chain of length $2N$ with rotational $\ZZ/N$-symmetry and reflection symmetry that permutes sites $i$ and $(2N-i)$. Let us denote by $J$ the anti-unitary operator on the total Hilbert space that implements reflection and complex conjugation. Then pure reflection positive states are represented by $J$-invariant vectors which span a {\it positive self-dual cone} in the Hilbert space. In particular, any two vectors in the interior of this cone must have a positive overlap. Since any two states which have different angular momenta are orthogonal, only one of them can lie in the interior of this cone, and therefore only one of them can be strictly reflection positive. In fact, as we show in Proposition \ref{prop:AngularMomenumofRP1dspinchain}, the angular momentum of a reflection positive state must vanish.

The same idea can be applied to states on an infinite two-dimensional spin system. While we do not have a canonical measure of the angular momentum, we can still consider the span of reflection positive vectors in the associated GNS Hilbert space. Only vectors with a particular angular momentum can lie in the interior of this cone. We show that any invertible state admits a reflection positive $\ZZ/N$-rotationally-invariant representative. Through these states, we define the canonical $\ZZ/N$-charge for twist defects.

{\bf Disclaimer.} This paper is concerned only with bosonic spin systems. The case of fermionic degrees of freedom will be analyzed in a separate publication.
\\

\begin{figure}
\centering




\begin{tikzpicture}[scale=.3]

\foreach \i in {0,...,7}{
  \pgfmathsetmacro{\y}{-5 + 0.5*\i} 
  \filldraw[orange, ultra thick, fill=orange!5]
    (-5,\y) -- (-1,\y+4) -- (9,\y+4) -- (5,\y);
}

\foreach \i in {0,...,6}{
  \pgfmathsetmacro{\y}{-5 + 0.5*\i}
  \draw[orange, ultra thick]
    (-5,\y)
    .. controls (-0.5,\y) ..
    (0,\y+0.20)
    .. controls (0.5,\y+0.40) ..
    (5,\y+0.50);
}

\draw[orange, ultra thick]
  (-5,-1.5)
  .. controls (-0.5,-1.5) ..
  (0,-2.25)
  .. controls (0.5,-5.0) ..
  (5,-5.0);

\draw[orange, ultra thick]
  (-5,-4.5) .. controls (-0.5,-4.5) .. (0,-4.25) .. controls (0.5,-4.0) .. (5,-4.0);
\draw[orange, ultra thick]
  (-5,-5.0) .. controls (-0.5,-5.0) .. (0,-4.75) .. controls (0.5,-4.5) .. (5,-4.5);

\filldraw[orange] (1.85,0.35) circle (5pt);


\end{tikzpicture}
\hspace{.4cm}
\begin{tikzpicture}[scale=.4]


\filldraw[fill=orange!5, draw=orange, ultra thick] (0,0) circle (3.5cm);

\foreach \k in {0,...,7} {
  \coordinate (P\k) at ({3.5*cos(360*\k/8)}, {3.5*sin(360*\k/8)});
}

\foreach \k in {0,...,3} {
  \pgfmathtruncatemacro{\opp}{\k+4}
  \draw[orange, thick, dashed] (P\k) -- (P\opp);
}

\filldraw [orange] (0,0) circle (5pt);

\draw[->, thick] (-9,0) -- (-5,0) node[midway, above] {$z \to z^{1/N}$};;

\end{tikzpicture}
\caption{A twist defect state on an $N$-sheeted branched cover of a complex plane can be transformed into a $\ZZ/N$-rotationally invariant state by a map $z \to z^{1/N}$, $z \in \CCC$.}
\label{fig:TwistDefect2RotationallyInvariantState}
\end{figure}
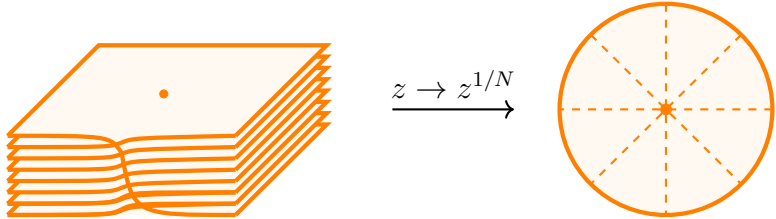

\subsection{Organization of the paper}
This paper is organized as follows.

In Section \ref{sec:generalities}, we give main definitions and describe a mathematical model for invertible phases on $\RR^2$. We show that in this model, states obtained by orientation-preserving homeomorphisms are in the same phase, while orientation-reversing homeomorphisms change the phase to its inverse. We introduce the notion of reflection positivity and prove that any two-dimensional invertible phase admits a reflection positive representative. We also prove that the complex conjugate state provides an inverse.

In Section \ref{sec:AngularMomentum}, we analyze systems with discrete $\ZZ/N$-rotational invariance. We introduce invertible phases protected by $\ZZ/N$-rotational symmetry and prove that reflection positive rotationally invariant states have a canonical angular momentum.

In Section \ref{sec:index}, we define an index for two-dimensional invertible phases by refining the construction from \cite{sopenko2024}. 

In Appendix \ref{app:canonicalpurification}, we review the construction of the canonical purification for an arbitrary $C^*$-algebra. We then consider systems with rotational symmetry and analyze the constraint imposed by reflection positivity on their properties. We describe some elementary applications for finite one-dimensional quantum systems.
\\

\noindent
{\bf Acknowledgements:} I am grateful to Anton Kapustin for his feedback on the draft and to Alexei Kitaev for discussions that brought my attention to the importance of analyzing the behavior of invertible phases under complex conjugation. I would also like to thank Yichul Choi, Dan Freed, Ryohei Kobayashi, Abhinav Prem, Daniel Ranard, Nati Seiberg, Sahand Seifnashri, Shu-Heng Shao, and Edward Witten for helpful conversations. This work is supported by NSF Grant PHY-2207584 and the Ambrose Monell Foundation.
\\

\noindent
{\bf Data availability statement:}
Data sharing not applicable to this article as no datasets were generated or analysed during the current study.
\\

\noindent
{\bf Declarations:}
The authors have no competing interests to declare that are relevant to the content of this article.
\\

\section{Invertible phases of 2d quantum spin systems} \label{sec:generalities}

\subsection{Quantum spin systems}

\begin{definition}
A spin system on a topological space $X$ is defined by the following data:
\begin{itemize}
    \item A countable set $\Lambda$ called a lattice and a map $\lambda: \Lambda \to X$ with locally finite image. The elements of $\Lambda$ are called sites, and $\lambda(j)$ is the location of a site $j \in \Lambda$ on $X$.
    \item A map $n:\Lambda \to \NN$ and an assignment of a finite dimensional Hilbert space $\hilb{V}_j$ of dimension $n(j)$ for each $j \in \Lambda$.
    \item For a finite subset $\Gamma \subset \Lambda$, a $C^*$-algebra $\cstar{A}_{\Gamma}$ of endomorphisms of $\bigotimes_{j \in \Gamma} \CV_{j}$.
    \item A (uniformly hyperfinite) $C^*$-algebra $\cstar{A}$ which is defined as an operator norm completion of the direct limit $\varinjlim_{\Gamma \subset \Lambda} \cstar{A}_{\Gamma}$ over natural inclusions $\cstar{A}_{\Gamma} \subset \cstar{A}_{\Gamma'}$ for $\Gamma \subset \Gamma'$.    
\end{itemize}
For $j \in \Lambda$, we let $\cstar{A}_j := \cstar{A}_{\{j\}}$. For a subset $\Lambda' \subset \Lambda$ we let $\cstar{A}_{\Lambda'}$ be the operator norm completion of the direct limit $\varinjlim_{\Gamma \subset \Lambda'} \cstar{A}_{\Gamma}$ over subsets of $\Lambda'$. For a subset $Y \subset X$, we let $\cstar{A}_{Y}$ be $\cstar{A}_{\Lambda'}$, where $\Lambda'$ is the preimage of $Y$ under $\lambda$. By abuse of notation, we often denote the whole data $(\Lambda, \lambda, n, \cstar{A})$ by a single letter $\cstar{A}$, and say that $\cstar{A}$ is a spin system on $X$.
\end{definition}

Two spin systems are isomorphic if we have a bijection $\varphi:\Lambda_1 \xrightarrow{\sim} \Lambda_2$ between their lattices such that $\lambda_1 = \lambda_2 \circ \varphi$ and $n_1 = n_2 \circ \varphi$. By an automorphism of a spin system $\cstar{A}$ we mean a linear $*$-automorphism of the corresponding $C^*$-algebra. By an anti-linear automorphism of $\cstar{A}$ we mean an anti-linear $*$-automorphism. By an automorphism of a $C^*$-algebra (either linear or anti-linear) we always mean a $*$-automorphism. For a given spin system, we say that $\alpha$ is an automorphism on $U \subset X$ if it acts trivially on observables $a \in \cstar{A}_{U^c}$, where $U^c$ is the complement of $U$. We say that a (possibly anti-linear) automorphism $\alpha$ of a spin system $(\Lambda,\lambda,n,\cstar{A})$ is {\it on-site} if for any $j\in \Lambda$ and $x \in \cstar{A}_j$, we have $\alpha(x) \in \cstar{A}_j$. For a subset $Y \subset X$ and an on-site automorphism $\alpha$, we denote by $\alpha_Y$ the automorphism that acts as $\alpha$ on $\cstar{A}_Y$ and trivially on $\cstar{A}_{Y^c}$.

We have a monoidal structure on spin systems: given two spin systems $(\Lambda_1, \lambda_1, n_1, \cstar{A}_1)$ and $(\Lambda_2, \lambda_2, n_2, \cstar{A}_2)$ on the same space $X$, we can form a new system $\cstar{A}_1 \otimes \cstar{A}_2$ by taking the union of their lattices $\Lambda = \Lambda_1 \cup \Lambda_2$ with $\lambda$ and $n$ being induced by $\lambda_{1,2}$, $n_{1,2}$. We call $\cstar{A}_1 \otimes \cstar{A}_2$ a {\it stack} of $\cstar{A}_1$ and $\cstar{A}_2$. 

A state on a spin system $\cstar{A}$ is a positive normalized linear functional $\cstar{A} \to \CCC$ on the $C^*$-algebra $\cstar{A}$. By abuse of notation, we often denote the whole data $(\Lambda, \lambda,n,\cstar{A},\psi)$ of a spin system and a state on it by a single letter $\psi$. If $\psi_1$, $\psi_2$ are two states on spin systems $\cstar{A}_1$, $\cstar{A}_2$ on the same space $X$, by $\psi_1 \otimes \psi_2$ we mean the tensor product state on $\cstar{A}_1 \otimes \cstar{A}_2$. For a subset $U \subset X$ and a state $\psi$ on a spin system $\cstar{A}$, by $\psi|_Y$ we denote the restriction of a state $\psi$ to $\cstar{A}_Y$. A state on a spin system $(\Lambda,\lambda,n, \cstar{A})$ is called a product state if it is a tensor product of states on $\{\cstar{A}_j\}_{j \in \Lambda}$. A pure product state is called {\it trivial}. We say that $\tilde{\psi}$ is a {\it trivial extension} of a state $\psi$ if $\tilde{\psi} = \psi \otimes \psi_0$ for some trivial state $\psi_0$.

In the paper, we extensively use the notion of a {\it factorial} state and {\it quasi-equivalence} of states. We remind the reader that, as was shown in \cite{powers1967representations}, for a uniformly hyperfinite $C^*$-algebras $\cstar{A}$ that is the limit of matrix algebras $\cstar{M}_{1} \subset \cstar{M}_{2} \subset ... \,$, $\cstar{M}_i \cong M_{l_i}(\CCC)$, a state $\omega$ is factorial if and only if it has decay at infinity, i.e. for any $x \in \cstar{M}_r$ and $\eps>0$ there exists $s > r$ such that for any $y \in \cstar{A}$ that commutes with $\cstar{M}_s$, we have 
$$|\omega(x y) - \omega(x) \omega(y)| < \eps \|x\| \|y\|.
$$
In particular, all pure states on spin systems and their restrictions to subsystems are factorial. Two factorial states $\omega_1, \omega_2$ are quasi-equivalent if and only if they coincide at infinity, i.e. for any $\eps>0$ there exists $r$, such that for any $x \in \cstar{A}$ that commutes with $\cstar{M}_r$ we have 
$$
|\omega_1(x) - \omega_2(x)| < \eps \|x\|.$$

In the following we stick to $X = \RR^2$. We always assume that $\RR^2$ is oriented.

\subsection{Good conical covers}

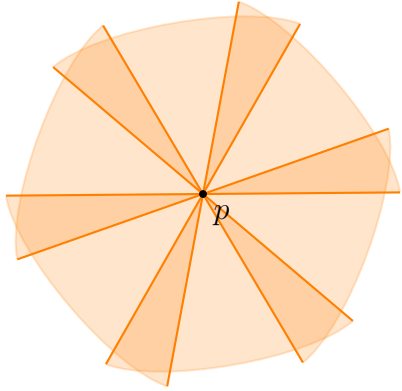
\begin{figure}
\centering
\begin{tikzpicture}[scale=0.5,rotate=10]
    \def\num_cones{6}
    \def\angle_step{360/\num_cones}
    \def\distance{4}   
    \def\width{3.3}    
    \def\bow{0.8}      

    \foreach \i in {1,...,\num_cones} {
        \begin{scope}[rotate = \i * \angle_step]
            \fill[orange, fill opacity=0.2]
                (0,0)
                -- (-\width, \distance)
                .. controls (-0.6*\width, \distance + \bow) and (0.6*\width, \distance + \bow) ..
                   (\width, \distance)
                -- cycle;

            \draw[orange, thick] (0,0) -- (-\width, \distance);
            \draw[orange, thick] (0,0) -- (\width, \distance);

            \draw[orange, opacity=0.2, thick]
                (-\width, \distance)
                .. controls (-0.6*\width, \distance + \bow) and (0.6*\width, \distance + \bow) ..
                   (\width, \distance);
        \end{scope}
    }

    \filldraw [black] (0,0) circle (2.5pt) node[below right] {$p$};
\end{tikzpicture}

\caption{An example of a good conical cover of $\RR^2$ that consists of six cones.}
\label{fig:goodcover}
\end{figure}

An open interval $I \subset S^1$ on a circle $S^1$ at infinity of $\RR^2$ and a point $p \in \RR^2$ defines an open {\it cone} that consists of all the points of $\RR^2 \setminus \{p\}$ whose projection to $S^1$ through the point $p$ lies inside $I$. We say that $I$ is the {\it base} of the cone, and $p$ is the {\it apex}.
\begin{definition}
A {\it good conical cover} of $\RR^2$ is a finite collection of cones $C = \{U_a\}_{a \in J}$, such that they all have the same apex and their bases $\{I_a\}_{a \in J}$ form a good open cover of $S^1$ such that $I_a \cap I_b$ is strictly smaller than both $I_a$, $I_b$ for any $a,b \in J$. 
\end{definition}

A good conical cover $C'$ is said to be {\it finer} than a good conical cover $C$ if for every $U' \in C'$ there exists $U \in C$ such that $U' \subseteq U$.

\begin{definition}
Let $\cstar{A}$ be a spin system on $\RR^2$, and let $C = \{U_a\}_{a \in J}$ be a good conical cover. We say that an automorphism $\alpha$ of $\cstar{A}$ is {\it $C$-local} if it can be written as $\l \prod_{a \in J_0} \alpha_{a} \r \l \prod_{a \in J_1} \alpha_{a} \r \Ad_{u}$, where $J_0$, $J_1 \subset J$ are disjoint subsets of non-overlapping cones with $J_0 \cup J_1 = J$, $\alpha_a$ is an automorphism on $U_a$, and $\Ad_{u}$ is an inner automorphism $\Ad_u(x):= u x u^*$, $x \in \cstar{A}$. We say that $\alpha$ is {\it strictly $C$-local} if the same decomposition exists with $u=1$.
\end{definition}

Loosely speaking, for a sufficiently fine cover $C$, a $C$-local automorphism acts independently on disjoint cones up to an overall inner conjugation. Note that if a good conical cover $C'$ is sufficiently finer than a good conical cover $C$, i.e. if the cones of $C'$ have the bases much smaller compared to the bases of the cones of $C$, then any $C'$-local automorphism is also $C$-local.

\begin{example} \label{example:HamiltonianEvolution}
An on-site automorphism gives an example of an automorphism that is $C$-local for any good conical cover $C$. More generally, the class of automorphisms obtained by integration of a derivation of $\cstar{A}$ exhibiting the Lieb-Robinson bound is also $C$-local for any $C$ \cite{naaijkens2022split}. Such automorphisms physically correspond to adiabatic evolutions by local Hamiltonians with sufficiently fast decay of interactions. In particular, any finite depth quantum circuit that consists of unitary gates with uniformly bounded support is $C$-local for any $C$.
\end{example}

\subsection{Invertible phases}

Topological phases correspond to equivalence classes of pure states of spin systems with respect to a certain equivalence relation. There are two main physical aspects this equivalence relation is supposed to capture: two states should be equivalent if they are related by stacking with trivial states (i.e., we can always add decoupled degrees of freedom without changing the phase) and if they are related by a continuous family of automorphisms of the algebra of observables which preserve some notion of locality (i.e., one state can be obtained from another by an adiabatic locality-preserving evolution).

There are different ways to define locality of automorphisms. We use the following formalization

\begin{definition} \label{def:Phases}
Two pure states $\psi_1$, $\psi_2$ on $\RR^2$ are {\it in the same phase} if for any good conical cover $C$ there exist states $\tilde{\psi}_1$, $\tilde{\psi}_{2}$ on a spin system $\cstar{A}$ such that $\tilde{\psi}_1$ is a trivial extension of $\psi_1$, $\tilde{\psi}_2$ is a trivial extension of $\psi_2$, and there exists a $C$-local automorphism $\alpha$ of $\cstar{A}$ such that $\tilde{\psi}_1 = \tilde{\psi}_2 \alpha$. Being in the same phase is an equivalence relation. Equivalence classes are called {\it phases}. If $\psi$, $\phi$ are representatives of two phases, then a {\it stack} of these phases is a phase represented by $\psi \otimes \phi$. {\it Stacking} defines a commutative monoidal operation on the set of phases, with a unit being the phase represented by a trivial state. The unit is also called {\it a trivial phase}. Invertible elements of the monoid are called {\it invertible phases} and form an Abelian group $\InvPhases$. The states representing them are called invertible. A state $\psi'$ is called an inverse for an invertible state $\psi$ if $\psi \otimes \psi'$ is in a trivial phase.
\end{definition}

\begin{remark}
We note that the definition we use is, a priori, weaker than what is usually proposed in the literature (see e.g. \cite{bachmann2012automorphic,naaijkens2022split}). Indeed, any local Hamiltonian evolution with sufficiently fast decay of interactions provides an example of an automorphism that is $C$-local for any good conical cover $C$ \cite{naaijkens2022split}. Therefore, the invariants of phases we construct would be automatically invariants for a stronger equivalence relation.
\end{remark}

The following lemma shows that in order to check whether two states are in the same phase, it is enough to consider conical covers with a fixed apex.

\begin{lemma} \label{lma:movingconicalcover}
Let $\psi_1, \psi_2$ be invertible states on $\RR^2$, and let $p \in \RR^2$. Suppose for any good conical cover $C$ with an apex at $p$ there exist states $\tilde{\psi}_1$, $\tilde{\psi}_{2}$ on a spin system $\cstar{A}$ such that $\tilde{\psi}_1$ is a trivial extension of $\psi_1$, $\tilde{\psi}_2$ is a trivial extension of $\psi_2$, and there exists a $C$-local automorphism $\alpha$ of $\cstar{A}$ such that $\tilde{\psi}_1 = \tilde{\psi}_2 \alpha$. Then the states $\psi_1$, $\psi_2$ are in the same phase.
\end{lemma}
\begin{proof}
Let $\cstar{B} = M_n(\CCC) \otimes \cstar{B}'$ be a UHF algebra, $\alpha$ be an automorphism of $\cstar{B}$, and $\{e_{ij}\}_{i,j =1}^N$ be matrix units which span $M_n(\CCC) \otimes 1$. Since $\cstar{B}$ has a unique trace, the projections $e_{11}$, $\alpha(e_{11})$ are Murray–von Neumann equivalent and there exists a partial isometry $w \in \cstar{B}$ satisfying $w^* w = e_{11}$, $w w^* = \alpha(e_{11})$. Then $u=\sum_{i=1}^n \alpha(e_{i 1}) w e_{1 i}$ is unitary, and $u e_{ij} u^* = \alpha(e_{ij})$. Therefore $\alpha$ can be represented as a composition of $\Ad_u$ and an automorphism $\alpha' = \Ad_{u^*} \alpha$ that acts trivially on $M_n(\CCC) \subset \cstar{B}$.

Let $C$ be a good conical cover with an apex different from $p$. We can always find a good conical cover $C'$ with an apex at $p$ and a bijection between the cones $U_a \to U'_a$ of $C$ and $C'$ such that $U_a$ is almost contained in $U'_a$, i.e. the complement of the intersection $U_a \cap U'_a$ in $U_a$ is a bounded region. As explained in the previous paragraph, any automorphism of the algebra on $U_a$ can be represented as a composition of an automorphism of the algebra on $U_a \cap U'_a$ and an inner automorphism. It follows that any $C$-local automorphism is also $C'$-local.
\end{proof}

\subsection{Spatial deformations}

Given a spin system $(\Lambda,\lambda,n,\cstar{A})$ and a homeomorphism $f:\RR^2 \to \RR^2$, we may define a new spin system $f_* \cstar{A}$ by postcomposing $\lambda$ and $n$ with $f$. In this case, given a state $\psi$ on $\cstar{A}$, we get a state $f_* \psi$ on $f_* \cstar{A}$ in a natural way.

We say that a homeomorphism $f$ of $\RR^2$ is {\it good} if it has the form $f = h^{-1} \circ g \circ h$, where $h$ is the radial contraction map $\RR^2 \to D^2 \setminus \p D$, $x \to \tanh(|x|) x$, and $g$ is the homeomorphism of $D^2\setminus \p D^2$ induced by a homeomorphism of a closed unit disk $D^2$. Loosely speaking, it means that $f$ has a good asymptotic behavior at infinity.

For $\eps>0$, we say that a map $f:\RR^2 \to \RR^2$ is {\it asymptotically $\eps$-local} if for some point $p\in\RR^2$, a bounded region $Y \subset \RR^2$ and any $x \in (\RR^2 \setminus Y)$ the angular distance between $f(x)$ and $x$ with respect to $p$ is less than $\eps$.
\begin{prop} \label{prop:InvarianceUnderDiffeomorphisms} 
Let $f$ be a good orientation-preserving homeomorphism of $\RR^2$. Then for any invertible state $\psi$ on $\RR^2$, the states $\psi$, $f_* \psi$ are in the same phase.
\end{prop}
\begin{proof}
Let $g$ be an orientation-preserving homeomorphism of a closed unit disk $D^2$, such that $f = h^{-1} \circ g \circ h$ for a radial contraction map $h:\RR^2 \to D^2\setminus\p D^2$ above. By Alexander's trick and a well-known fact that any two orientation-preserving homeomorphisms of $S^1$ are isotopic, $g$ is isotopic to the identity map. Since $D^2 \times [0,1]$ is compact, for any $\eps>0$, there exist homeomorphisms $g_1,...,g_N$ of $D^2$ such that $g = g_N \circ ... \circ g_1$ and for any $x \in D^2$ we have $|g_i(x)-x| < \eps$. The corresponding homeomorphisms $f_i = h^{-1} \circ g_i \circ h$ of $\RR^2$ are good and asymptotically $\eps$-local.

Let $C$ be a good conical cover. For a sufficiently small $\eps > 0$, we can always find a finer good conical cover $C'$ such that if $\psi \otimes \psi'$ is related to a trivial state by a $C'$-local automorphism, then the state $(\psi \otimes f_{i*} \psi')$ is related to a trivial state by a $C$-local automorphism for any asymptotically $\eps$-local homeomorphism $f_i$.

Choose homeomorphisms $f_1,...,f_N$ for this value of $\eps$ and let $C'$ be the corresponding cover from the previous paragraph. Let $\psi_0 = \psi$, $\psi_i = (f_i \circ ... \circ f_1)_* \psi$, and let $\psi'_i$ be inverses of $\psi_i$ such that $\psi'_i \otimes \psi_i$ are related to a trivial state by a $C'$-local (and therefore by a $C$-local) automorphism. Since $f_i$ are $\eps$-local, the states $\psi_i \otimes \psi'_{i+1}$ are related to a trivial state by a $C$-local automorphism. Thus, there are $C$-local automorphisms $\alpha$, $\alpha'$, such that $(\psi_0 \otimes \psi'_1 \otimes \psi_1 \otimes \psi'_2 \otimes ... \otimes \psi'_n \otimes \psi_n ) \alpha$ is a trivial extension of $\psi_0=\psi$ and $(\psi_0 \otimes \psi'_1 \otimes \psi_1 \otimes \psi'_2 \otimes ... \otimes \psi'_n \otimes \psi_n ) \alpha'$ is a trivial extension of $\psi_n = f_* \psi$. We obtained trivial extensions of $\psi$ and $f_* \psi$ which are related by a $C$-local automorphism for any good conical cover $C$.
\end{proof}

\begin{lemma} \label{lma:stateonahalfplane}
Let $\psi_0$, $\psi_1$ be invertible states such that for a cone $V$, the states $\psi_0|_{V}$, $\psi_1|_{V}$ are quasi-equivalent. Then the states $\psi_0, \psi_1$ are in the same phase.
\end{lemma}
\begin{proof}
Since we can stack both states with an inverse for $\psi_0$, without loss of generality we can assume that $\psi_0$ is in a trivial phase.

Let $C$ be a good conical cover with the same apex as $V$. By Lemma \ref{lma:movingconicalcover}, it is enough to show that there exists a $C$-local automorphism that transforms $\psi_1$ into a trivial state, possibly after a trivial extension.

Let $C' = \{U_a\}_{a \in J}$ be a good conical cover, and let $f$ be a good orientation-preserving homeomorphism of $\RR^2$ such that $f(V^c)$ overlaps only with a cone $U$ of $C'$. The state $f_* \psi_0$ is in a trivial phase by Proposition \ref{prop:InvarianceUnderDiffeomorphisms}. Hence, possibly after a trivial extension of $\psi_0$, there exists a $C'$-local automorphism $\alpha = \prod_{a \in J_0} \alpha_a \prod_{a \in J_1} \alpha_a \Ad_{u}$ such that $U=U_a$ for $a \in J_1$ and $f_* \psi_0 \alpha$ is trivial. Since $f_*\psi_1|_{U^c}$ is quasi-equivalent to $f_* \psi_0|_{U^c}$, the state $(f_*\psi_1 \alpha)|_{U^c}$ is quasi-equivalent to a trivial state. It follows that $(f_*\psi_1 \alpha)|_U$ is quasi-equivalent to a pure state, and therefore by transitivity of the action of automorphisms on the set of pure states, the state $f_*\psi_1 \alpha \beta$ is quasi-equivalent to a trivial state for some automorphism $\beta$ on $U$. Thus, by Kadison transitivity theorem, $f_* \psi_1$ is related to a trivial state by a $C'$-local automorphism. By Proposition \ref{prop:InvarianceUnderDiffeomorphisms}, $f_*\psi_1$ and $\psi_1$ are also related by a $C'$-local automorphism, possibly after a trivial extension. Since $C'$ can be chosen fine enough so that a composition of any two $C'$-local automorphisms is $C$-local, $\psi_1$ can be related to a trivial state by a $C$-local automorphism.
\end{proof}

\begin{prop} \label{prop:psiisometry}
Let $f$ be a good orientation-reversing homeomorphism of $\RR^2$. Then for any invertible state $\psi$ on $\RR^2$, the states $\psi$, $f_* \psi$ are in the opposite phases.
\end{prop}
\begin{proof}
Any good orientation-reversing homeomorphism is a composition of a good orientation-preserving homeomorphism and a reflection in some line. Therefore, by Proposition \ref{prop:InvarianceUnderDiffeomorphisms}, it is enough to consider the case of $f$ being a reflection.

Let $l$ be a line separating $\RR^2$ into two half-spaces $H_L$, $H_R$ that does not pass through the locations of the sites of the lattice and let $f$ be a reflection in $l$. Let $\phi$ be a state obtained from $\psi$ by folding the plane along $l$ so that $H_R$ is mapped to $H_L$ as in the reflection in $l$. Since $\phi|_{H_R}$ is trivial, by Lemma \ref{lma:stateonahalfplane}, $\phi$ is in a trivial phase. Let $U$ be a cone in the interior of $H_L$, and let $\omega = \psi \otimes f_* \psi$. By Lemma \ref{lma:stateonahalfplane}, it is enough to show that $\omega|_{U}$ is quasi-equivalent to $\phi|_{U}$. 

Suppose $\psi$ is in a trivial phase, and let us choose a good conical cover $C$ with sufficiently thin (compared to $U$) cones. If $\psi$ can be obtained from a trivial state by a strictly $C$-local automorphism, then $\omega|_{U}$ and $\phi|_{U}$ are the same. If $\psi$ is unitarily equivalent to a state $\tilde{\psi}$ like this, then $\omega$ and $\phi$ would be unitarily equivalent to states $\tilde{\omega}$ and $\tilde{\phi}$ constructed from $\tilde{\psi}$ in the same way as $\omega$ and $\phi$ are constructed from $\psi$. Therefore, $\omega|_{U}$ and $\phi|_{U}$ are quasi-equivalent.

If $\psi$ is in a nontrivial invertible phase, then we can choose an inverse $\psi'$ and construct $\phi'$ and $\omega'$ from $\psi'$ in the same way as $\phi$ and $\omega$ are constructed from $\psi$. Then, quasi-equivalence of $(\omega \otimes \omega')|_U$ and $(\phi \otimes \phi')|_U$ implies quasi-equivalence of $\omega|_U$ and $\phi|_U$.

\end{proof}

\begin{remark}
While we stick to $X=\RR^2$ in this paper, the definitions of a good conical cover $C$, a $C$-local automorphism and an invertible phase can be generalized to the case of $X = \RR^d$ in a straightforward way. In the definition of a $C$-local automorphism, we just need to decompose the cover into $d$ subsets of non-overlapping cones instead of $2$ for $d=2$, and consider automorphisms which are compositions of $d$ "layers", each of which is a product of automorphisms on non-overlapping cones. Using the same methods and the fact that any orientation-preserving homeomorphism of $D^d$ is isotopic to the identity map, we can deduce that orientation-preserving good homeomorphisms do not change the phase of an invertible state, while orientation-reversing change it to the opposite.
\end{remark}

\subsection{Reflection positivity and complex conjugation}

\begin{prop} \label{prop:psicomplexconjugation}
Let $\psi$ be an invertible state on $\RR^2$. Let $\tau$ be an on-site anti-linear automorphism of the corresponding spin system $\cstar{A}$. Then the state $\bar{\psi}$ defined by $\bar{\psi}(a) := \overline{\psi(\tau(a))}$ is an inverse of $\psi$.
\end{prop}

In order to prove this proposition, we use the main ingredient of this paper: canonical purification. We refer the reader to Appendix \ref{app:canonicalpurification} for necessary details and the terminology used herein.

\begin{definition}
Let $\cstar{A}$ be a spin system and let $l$ be a line separating $\RR^2$ into two half-planes $H_L, H_R$ such that $l$ does not pass through the locations of the sites of the lattice and $\cstar{A}$ is invariant under reflection in $l$. By a {\it $CRT$-symmetry in $l$} we mean an anti-linear involutive automorphism $j$ such that for any $x \in \cstar{A}_k$ we have $j(x) \in \cstar{A}_{k'}$, where $k'$ is a reflection of the site $k$ in $l$. We say that a state $\psi$ on $\cstar{A}$ is {\it (strictly) reflection positive} with respect to $j$ if it is (strictly) $j$-positive on $\cstar{A}$ with respect to factorization $\cstar{A} \cong \cstar{A}_{H_L} \otimes \cstar{A}_{H_R}$ (as defined in Appendix \ref{app:canonicalpurification}).
\end{definition}

\begin{lemma} \label{lma:psicanonicalpurification}
Let $\cstar{A}$ be a spin system and $l$ be a line separating $\RR^2$ into two half-planes $H_L$, $H_R$ that does not pass through the locations of the sites of the lattice. Let $\tilde{\cstar{A}}$ be a spin system obtained by reflecting $\cstar{A}_{H_L}$ in the line $l$ and equipped with a $CRT$-symmetry $j$ in $l$. Let $\psi$ be an invertible state on $\cstar{A}$, and let $\phi$ be the state obtained by canonical purification of $\psi|_{H_L}$ with respect to $j$. Then $\phi$ is invertible and is in the same phase as $\psi$.
\end{lemma}

\begin{proof}
It is enough to show that $\phi$ is invertible. The fact that it is in the same phase as $\psi$ follows from Lemma \ref{lma:stateonahalfplane}. Without loss of generality, we can also assume that $\psi$ is in a trivial phase, because we can always stack it with an inverse $\psi'$ and apply the same procedure to construct a state $\phi \otimes \phi'$.

Let $p \in l$ be a point on $l$, and let $l_+$, $l_-$ be the half-lines separated by $p$. Let $C$ be a good conical cover with an apex at $p$ and such that each $l_{\pm}$ is contained inside a single cone $U_{\pm}$ in $C$. We let $U'_{\pm}$ be the reflection invariant cones which have the same intersection with $H_L$ as $U_{\pm}$. We can construct a new good conical cover $C'$ that consists of $U'_{+}$, $U'_{-}$, the cones of the conical cover $C$ inside $H_L$, and the cones in $H_R$ obtained by reflection in $l$ of the latter. We say that $C'$ is reflection invariant. Given Lemma \ref{lma:movingconicalcover} and the fact that for any good conical cover we can find a finer reflection invariant cover, it is enough to show that if $\psi$ is related to a trivial state by a $C'$-local automorphism, then so is $\phi$.

Suppose there exist $C'$-local automorphisms $\alpha_L$, $\alpha_R$ on $H_L$, $H_R$, automorphisms $\alpha_+, \alpha_-$ on $U'_+$, $U'_-$, and an inner automorphism $\Ad_u$ such that $\psi \alpha_L \alpha_R \alpha_+ \alpha_- \Ad_{u}$ is a trivial state.
Let $\tilde{\psi} = \psi \alpha_L \alpha_R$ and let $\tilde{\phi}$ be the canonical purification of the restriction of $\tilde{\psi}$ to $H_L$. The state $\tilde{\psi}|_{U'_+} \otimes \tilde{\psi}|_{U'^{c}_+}$ is quasi-equivalent to $\tilde{\psi}$. Therefore, by Theorem \ref{thm:canonicalpurificationquasiequivalence}, the state $\tilde{\phi}|_{U'_+} \otimes \tilde{\phi}|_{U'^{c}_+}$ is quasi-equivalent to $\tilde{\phi}$. It follows that $\tilde{\phi}|_{U'_+}$ is quasi-equivalent to a pure state, and by transitivity of the action of automorphisms on the set of pure states, there exist an automorphism $\beta_+$ such that $\tilde{\phi}|_{U'_+} \beta_+$ is quasi-equivalent to a trivial state. Similarly, there exists an automorphism $\beta_-$ such that $\tilde{\phi}|_{U'_-} \beta_-$ is quasi-equivalent to a trivial state. Thus, $\tilde{\phi} \beta_+ \beta_-$ is quasi-equivalent to a trivial state. Since $\tilde{\phi} = \phi \alpha_L (j \alpha_L j)$, the state $\phi \alpha_L (j \alpha_L j) \beta_+ \beta_-$ is unitarily equivalent to a trivial state. Since $(j \alpha_L j)$ is $C'$-local on $H_R$, it follows from Kadison transitivity theorem that $\phi$ can be transformed into a trivial state by a $C'$-local automorphism.
\end{proof}

\begin{proof}[Proof of Proposition \ref{prop:psicomplexconjugation}]
Let us choose a line $l$ not passing through sites of the lattice, and construct a state $\phi$ in the same way as in Lemma \ref{lma:psicanonicalpurification} with a CRT-symmetry $j$. By the same lemma, $\phi$ is in the same phase as $\psi$. The state $\bar{\phi}$ defined via $\bar{\phi}(a) := \overline{\phi(\tau(a))}$ up to an action with an on-site automorphism coincides with the state obtained by a reflection in $l$. By Proposition \ref{prop:psiisometry}, $\bar{\phi}$ is an inverse for $\phi$, and by Lemma \ref{lma:stateonahalfplane} is in the same phase as $\bar{\psi}$. Thus, $\bar{\psi}$ is an inverse for $\psi$.
\end{proof}

\section{Twist defects and canonical angular momentum} \label{sec:AngularMomentum}

In this section, we analyze invertible states which are invariant under discrete $\ZZ/N$-rotational symmetry. We study invertible phases $\InvPhasesZN$ protected by this symmetry and show that the forgetful map $\InvPhasesZN \to \InvPhases$ is surjective with exactly $N$ preimages for any element of $\InvPhases$. These different preimages physically correspond to different $\ZZ/N$ angular momenta for a given $\ZZ/N$-invariant state.

We then show that only one of the preimages can have a reflection positive representative. In other words, reflection positivity provides a canonical choice of the angular momentum for a $\ZZ/N$-rotationally invariant state. Therefore, there is a canonical lift $\InvPhases \to \InvPhasesZN$.

Before discussing rotationally invariant states, we introduce twist defect states from \cite{sopenko2024} which are closely related.

\subsection{Twist defect states} \label{ssec:twistdefect}

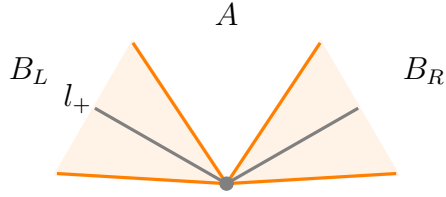
\begin{figure}
\centering
\begin{tikzpicture}[scale=.5]
\filldraw[color=orange!10, fill=orange!10, ultra thick] (0,0) -- (3.4641-1,2+3.4641/2) -- (3.4641+1,2-3.4641/2) -- cycle;
\filldraw[color=orange!10, fill=orange!10, ultra thick] (0,0) -- (-3.4641+1,2+3.4641/2) -- (-3.4641-1,2-3.4641/2) -- cycle;
\draw[orange, very thick] (0,0) -- (3.4641+1,2-3.4641/2);
\draw[orange, very thick] (0,0) -- (3.4641-1,2+3.4641/2);
\draw[orange, very thick] (0,0) -- (-3.4641-1,2-3.4641/2);
\draw[orange, very thick] (0,0) -- (-3.4641+1,2+3.4641/2);

\draw[gray, very thick] (0,0) -- (3.4641,2);
\draw[gray, very thick] (0,0) -- (-3.4641,2);

\node  at (1.5*1.7321*4/2,1.5*4/2) {$B_{R}$};
\node  at (-1.13*1.7321*4/2,1.13*4/2) {$l_{+}$};
\node  at (-1.5*1.7321*4/2,1.5*4/2) {$B_{L}$};
\node  at (0,1.5*3) {$A$}; 

\filldraw [gray] (0,0) circle (5pt);

\end{tikzpicture}
\caption{The cones $A, B_L, B_R$ and the half-line $l_+$. The boundaries of $A$ are depicted by gray lines.}
\label{fig:cones2}
\end{figure}

Let $\omega$ be an invertible state on $\RR^2$, and let $N$ be a positive integer. Let $\cstar{A}$ be the stack of $N$ copies of the spin system and let $\psi = \omega^{\otimes N}$ be the corresponding state. The symmetric group $S_N$ acts on $\cstar{A}$ via on-site automorphisms in a natural way and preserves the state $\psi$. Let $\rho$ be the automorphism corresponding to the generator of $\ZZ/N \subset S_N$ subgroup that is the cyclic permutation $\{1,2,...,N\} \to \{N,1,2,...,N-1\}$.

Let us pick a half-line $l_+$ with an origin at point $p$ that does not pass through locations of the sites of the lattice. Choose cones $A, B_L, B_R$ with an apex at $p$ and such that the boundary of $A$ that lies inside $B_L$ coincides with $l_+$ (see Fig. \ref{fig:cones2}). Let $\rho_A$ be the restriction of $\rho$ to $A$. As argued in Section 2.6 of \cite{sopenko2024}, the absence of non-trivial bosonic $\ZZ/N$-symmetry protected phases in 1d implies that we can effectively undo the action of $\rho_A$ on $\psi$ by automorphisms acting on $B_L$ and $B_R$:
\begin{lemma}
There is a trivial extension $\tilde{\omega}$ of $\omega$ and $\rho$-invariant automorphisms $\kappa_L$, $\kappa_R$ on $B_L$, $B_R$ such that $\tilde{\psi} \rho_A$ is unitarily equivalent to $\tilde{\psi} \kappa_L \kappa^{-1}_R$ for $\tilde{\psi} = \tilde{\omega}^{\otimes N}$.
\end{lemma}
\begin{proof}
Let $\phi = \psi \rho_A$. Let $\omega'$ be an inverse for $\omega$ on a system $\cstar{A'}$ and $\psi' = \omega'^{\otimes N}$. We consider a system $\cstar{A} \otimes \cstar{A}' \otimes \cstar{A}$. We assume that on this system, $\rho$ acts in a natural way by permuting the factors of $N$-th tensor products for both $\cstar{A}$ and $\cstar{A}'$. Let $\phi \otimes \psi'_0 \otimes \psi_0$ be a state on $\cstar{A} \otimes \cstar{A}' \otimes \cstar{A}$ for some trivial states $\psi_0$, $\psi'_0$. For a good conical cover $C$ with an apex at $p$, we can find a strictly $C$-local $\rho$-invariant automorphism $\alpha$, such that $(\psi'_0 \otimes \psi_0)\alpha$ is unitarily equivalent to $\psi' \otimes \psi$. It is enough to show that there exist $\rho$-invariant automorphisms $\kappa_L$, $\kappa_R$ on $B_L$, $B_R$ such that $(\phi \otimes \psi')\kappa^{-1}_L \kappa_R$ is unitarily equivalent to $(\psi \otimes \psi')$. Indeed, in this case, we can take small enough cones $B'_L$, $B'_R$ with the same apex in the interior of $B_L$, $B_R$, find the appropriate $\kappa_L$, $\kappa_R$ on $B'_L$, $B'_R$, respectively, and conjugate them with $\text{id}_{\cstar{A}} \otimes \alpha$ to get the desired $\kappa_L$, $\kappa_R$ from the statement of the lemma.

Let $C$ be a fine enough good conical cover with an apex at $p$ and $\alpha$ be the $N$-th tensor product of a strictly $C$-local automorphism relating $\omega \otimes \omega'$ to a trivial state up to unitary equivalence. By construction, $\alpha$ is $\rho$-invariant. Since $(\psi \otimes \psi')\alpha$ is unitarily equivalent to a trivial state, it is enough to show the existence of $\rho$-invariant automorphisms $\kappa_L$, $\kappa_R$ such that $(\phi \otimes \psi')\alpha \kappa^{-1}_L \kappa_R$ is unitarily equivalent to a trivial state by the same argument as in the previous paragraph.

Note that the restriction of $(\phi \otimes \psi')\alpha$ to the complement of $B_L \cup B_R$ is quasi-equivalent to a trivial state, while the restriction to $B_L \cup B_R$ is quasi-equivalent to the tensor product of the restrictions to $B_L$ and $B_R$. It follows that the restrictions to $B_{L,R}$ are quasi-equivalent to pure states. Since $H^2(\ZZ/N,\RR/\ZZ)$ is trivial, by \cite{ogata2019classification}[Theorem 1.11], there exists a $\rho$-invariant automorphism $\kappa_L$ on $B_L$, such that the restriction of $(\phi \otimes \psi')\alpha \kappa^{-1}_L$ to $B_L$ is quasi-equivalent to a trivial state. Similarly, there exists a $\rho$-invariant automorphism $\kappa_R$ on $B_R$, such that the restriction of $(\phi \otimes \psi')\alpha \kappa_R$ to $B_R$ is quasi-equivalent to a trivial state. It follows that $(\phi \otimes \psi')\alpha \kappa^{-1}_L \kappa_R$ is unitarily equivalent to a trivial state.

\end{proof}

By the lemma above, possibly after replacing $\omega$ with its trivial extension, we can always find $\rho$-invariant automorphisms $\kappa_L$, $\kappa_R$ on $B_L$, $B_R$ such that $\psi \rho_A$ is unitarily equivalent to $\psi \kappa_L \kappa^{-1}_R$. A state $\varphi$ of the form $\psi \kappa_{L} \Ad_v$ for a $\rho$-invariant inner automorphism $\Ad_v$ is called a {\it twist defect state}. Physically, it corresponds to a state obtained by cutting each of $N$ copies of the spin system along $l_+$ and regluing them according to the cyclic permutation (see Fig. \ref{fig:twistdefect}). This operation is performed by the automorphism $\kappa_L$.

While the definition of twist defect state depends on the automorphisms $\kappa_L, \kappa_R$ and the choice of cones $A$, $B_L$, $B_R$, it is easy to see that once we have fixed $l_+$, all twist defects states are unitarily equivalent.
\begin{lemma}
Let $A'$, $B_L'$, $B_R'$, $\kappa'_L$, $\kappa'_R$ be another choice of cones and automorphisms on $B'_L$, $B'_R$ for a given $l_+$ such that $\psi \rho_{A'}$ is unitarily equivalent to $\psi \kappa'_L \kappa'^{-1}_R$. Then $\psi \kappa_L$ and $\psi \kappa'_L$ are unitarily equivalent.    
\end{lemma}
\begin{proof}
Any admissible deformation of a single cone $B_L$, $B_R$, $A$ can be equivalently described by a change of automorphisms $\kappa_L$, $\kappa_R$. We can get any admissible triple of cones from $B_L$, $B_R$, $A$ by deforming them one by one. Therefore, it is enough to consider the case when $A'$, $B_L'$, $B_R'$ coincide with $A$, $B_L$, $B_R$.

Let $\phi$ be a state in a trivial phase. Let $\beta_L$, $\beta_R$ be automorphisms on $B_L$,$B_R$, respectively, such that $\phi \beta_L \beta_R$ is unitarily equivalent to $\phi$. We choose a sufficiently fine good conical cover $C$ and a $C$-local automorphism $\alpha$, such that $\phi \alpha$ is trivial. Since automorphisms $(\alpha^{-1}\beta_L \alpha)$ and $(\alpha^{-1}\beta_R \alpha)$ act on non-overlapping cones, unitary equivalence of $(\phi \alpha)(\alpha^{-1}\beta_L \alpha)(\alpha^{-1}\beta_R \alpha)$ and $\phi \alpha$ implies unitary equivalence of $(\phi \alpha)(\alpha^{-1}\beta_L \alpha)$ and $\phi \alpha$. Thus, $\phi \beta_L$ and $\phi$ are unitarily equivalent.

By taking $\phi = \psi \otimes \psi'$ for some inverse $\psi'$ of $\psi$ and $\beta_{L} = \kappa_{L} \kappa'^{-1}_{L}$, $\beta_{R} = \kappa^{-1}_{R} \kappa'_{R}$, we obtain unitary equivalence between $\psi (\kappa_L \kappa'^{-1}_L)$ and $\psi$.
\end{proof}

This lemma implies that for any cone $U$ that does not intersect $l_+$, the states $\varphi|_U$ and $\psi|_U$ are quasi-equivalent.

\subsection{Rotational invariance} \label{ssec:rotationalinvariance}

\begin{definition}
Let $\cstar{A}$ be a spin system on $\RR^2$. We say that $\cstar{A}$ has a {\it $\ZZ/N$-rotational symmetry} if it is equipped with a linear automorphism $\ups$ of order $N$ such that for any site $j$ and any $x \in \cstar{A}_j$ we have $\ups(x) \in \cstar{A}_{j'}$, where $j'$ is the image of the site $j$ under $2 \pi /N$ rotation around the origin consistent with the orientation of $\RR^2$.
\end{definition}

Let $\omega$ be an invertible state, and let $\varphi$ be a twist defect state for $\omega$ with the cut along a half-line $l_+$. Let us identify $\RR^2$ with the complex plane $\CCC$, and let us consider $N$-sheeted branched cover $\Sigma$ of $\RR^2$ with the cut along $l_+$. We can think of the twist defect state as a state on $\Sigma$. Without loss of generality, we can assume that the cut corresponds to the positive real axis. We let $f:\Sigma \to \RR^2$ be a map that corresponds to the map $z \to z^{1/N}$, $z \in \CCC$ (see Fig. \ref{fig:TwistDefect2RotationallyInvariantState}). Then given a twist defect state $\varphi$ on $\Sigma$, we can construct a $\ZZ/N$-rotationally invariant state $f_* \varphi$ with the rotation automorphism being induced by the automorphism $\rho$ that permutes the layers of $\varphi$.

We can choose an orientation-preserving homeomorphism $g:\RR^2 \to \RR^2$ such that $g_* \omega |_U$ is quasi-equivalent to $f_* \varphi$ for some cone $U$. By Lemma \ref{lma:stateonahalfplane} and Proposition \ref{prop:InvarianceUnderDiffeomorphisms}, $\omega$ and $f_* \varphi$ are in the same phase. We have proven the following

\begin{lemma} \label{lma:ZNInvRotations}
Any invertible phase has a $\ZZ/N$-rotationally invariant representative.
\end{lemma}

Let $\ups$ be the generator of a $\ZZ/N$ rotational symmetry of a spin system $\cstar{A}$. We say an inner $\ZZ/N$-rotationally invariant automorphism $\Ad_u$ has angular momentum $k \in \ZZ/N$ if $\ups(u) = e^{2 \pi i k/N} u$. For a $\ZZ/N$-rotationally invariant good conical cover $C = \{U_a\}_{a \in J}$, we say that a $\ZZ/N$-rotationally invariant $C$-local automorphism $\alpha$ has angular momentum $k \in \ZZ/N$ if it can be written as $\l \prod_{a \in J_0} \alpha_{a} \r \l \prod_{a \in J_1} \alpha_{a} \r \Ad_{u}$ for $\Ad_u$ with the angular momentum $k$ (where $J_0$, $J_1 \subset J$ are disjoint subsets of non-overlapping cones with $J_0 \cup J_1 = J$, $\alpha_a$ is as automorphism of $\cstar{A}_{U_a}$).

We define {\it $\ZZ/N$-rotational symmetry protected invertible phases} as equivalence classes of $\ZZ/N$-rotationally invariant invertible states on $\RR^2$ in the same way as invertible phases with the exception that now we require good conical covers to be $\ZZ/N$-rotationally invariant and the $C$-local automorphisms relating the states be $\ZZ/N$-rotationally invariant with vanishing angular momentum. Such equivalence classes form an abelian group which we denote by $\InvPhasesZN$.

By Lemma \ref{lma:ZNInvRotations}, the forgetful map $\InvPhasesZN \to \InvPhases$ is surjective. On the other hand, for any two $\ZZ/N$-rotationally invariant states $\psi_1$, $\psi_2$ in the same invertible phase and any $\ZZ/N$-invariant good conical cover $C$ we can always find a strictly $C$-local $\ZZ/N$-invariant automorphism $\alpha$ such that $\psi_1$ and $\psi'_2 = \psi_2 \alpha$ are related by an inner $\ZZ/N$-rotationally invariant automorphism $\Ad_v$. Since there are $N$ different possibilities for the angular momentum of $\Ad_v$, the map $\InvPhasesZN \to \InvPhases$ has exactly $N$ preimages.

\subsection{Rotational invariance and CRT-symmetry}

\begin{definition}
Let $\cstar{A}$ be a spin system on $\RR^2$. We say that $\cstar{A}$ has a {\it spatial $D_N$ symmetry} if it is equipped with a $\ZZ/N$-rotational symmetry $\ups$ and a $CRT$-symmetry $j$ satisfying $j \ups j = \ups^{-1}$, which together generate an action of a dihedral group $D_{N} \subset O(2)$. In this case, we let $\{ j_a \}_{a \in \ZZ}$ be $CRT$-symmetries in lines $\{l_a\}_{a \in \ZZ}$ defined by $j_0 = j$, $j_a = \ups^a j_0 = \ups j_{a-2} \ups^{-1}$, $j_{a+N} = j_a$.
\end{definition}

Suppose we have a spin system with a spatial $D_N$ symmetry generated by $\ups$ and $j$. We will say that a state is {\it (strictly) reflection positive with respect to the spatial $D_N$ symmetry} if it is $\ups$-invariant and (strictly) reflection positive with respect to $j$.

\begin{prop} \label{prop:existenceofRPstate}
Any invertible phase admits a representative which is strictly reflection positive with respect to a spatial $D_N$ symmetry.
\end{prop}

\begin{proof}

Let $\psi$ be an invertible state on $\RR^2$. Let us choose a good conical cover $C$ with cones having angles sufficiently small compared to $2 \pi/N$. Let $\beta$ be a strictly $C$-local automorphism such that $(\psi \otimes \psi') \Ad_u \beta$ is a trivial state for some inner automorphism $\Ad_u$ and an inverse $\psi'$ of $\psi$. We let $\phi = (\psi \otimes \psi')\Ad_{u}$ which is unitarily equivalent to $\psi \otimes \psi'$.

Let $\{l_a\}_{a \in \ZZ}$ be the lines for $CRT$-symmetries $\{j_a\}_{a \in \ZZ}$ and let us orient them so that the orientation vectors for $l_a$ and $l_{a+1}$ have angle $2 \pi/N$ between them. Let us sequentially perform the following $N$ operations with $\phi$. At $n$-th step, we take a restriction of the state to the half-plane to the right of the vector corresponding to $l_{n-1}$ and canonically purify it. Since $\phi \beta$ is trivial and $\beta$ is strictly $C$-local, the resulting state $\tilde{\phi}$ is reflection positive with respect to the spatial $D_{N}$ symmetry. If we apply the same procedure to the state $\psi \otimes \psi'$, by Theorem \ref{thm:canonicalpurificationquasiequivalence} we get a state $\tilde{\psi} \otimes \tilde{\psi}'$ which is unitarily equivalent to $\tilde{\phi}$. However, it might not be reflection positive with respect to the $D_N$ symmetry.

The state $\tilde{\phi}$ is exact and reflection positive with respect to any $j_a$. Hence, by Proposition \ref{prop:makingstrictlyRP}, we can find a pure state $\omega$ which is unitarily equivalent to $\tilde{\phi}$ and is strictly reflection positive with respect to the spatial $D_N$ symmetry. Let $\cstar{A}$ be the spin system on which the state $\tilde{\psi}$ is defined. The state $\omega|_{\cstar{A}}$ is quasi-equivalent to $\tilde{\psi}$ and is a normal state in the GNS representation of $\tilde{\psi}$ satisfying the conditions of Proposition \ref{prop:perronfrobenius}. Thus, there is a pure strictly reflection positive state unitarily equivalent to $\tilde{\psi}$, which by Lemma \ref{lma:psicanonicalpurification}, is in the same phase as $\psi$.
\end{proof}

In the same way as with $\ZZ/N$-rotational symmetry, a given invertible phase has many $D_N$-invariant representatives which differ by their angular momentum with respect to discrete rotations. The theorem below shows that if in addition we require that the state is reflection positive, then up to $\ZZ/N$-rotational symmetry-protected equivalence, there is only one such representative. Therefore, it provides a canonical lift $\InvPhases \to \InvPhasesZN$. We say that representatives of the lifted phase have a {\it canonical angular momentum}.

\begin{theorem} \label{thm:RPfixesAngularMomentum}
Among $N$ preimages of a given invertible phase under the forgetful map $\InvPhasesZN \to \InvPhases$, there is one and only one that admits a representative which is strictly reflection positive with respect to a spatial $D_{N}$ symmetry.
\end{theorem}

\begin{proof}
By Proposition \ref{prop:existenceofRPstate}, there exists a pure state $\psi$ representing a given invertible phase on a spin system with a spatial $D_{N}$ symmetry which is strictly reflection positive. Suppose there is another state like this $\phi$, and that $\phi$ and $\psi$ represent different elements of $\InvPhasesZN$. Then the state $\omega = \psi \otimes \bar{\phi}$ represents a trivial element of $\InvPhases$, but a nontrivial element of $\InvPhasesZN$. Thus, it is enough to show that any strictly reflection positive state in a trivial phase in $\InvPhases$ is also in a trivial phase in $\InvPhasesZN$.

Let $\psi$ be a state in a trivial phase on $\RR^2$ which is strictly reflection positive with respect to a spatial $D_N$ symmetry. Let us choose a $D_N$-invariant good conical cover $C$ with cones having angles sufficiently small compared to $2 \pi/N$. Let $\beta$ be a strictly $C$-local automorphism such that $\psi \Ad_u \beta$ is a trivial state for some $D_N$-invariant inner automorphism $\Ad_u$. We let $\phi = \psi \Ad_{u}$ which is unitarily equivalent to $\psi$.

We apply the same procedure to the state $\phi$ as in the proof of Proposition \ref{prop:existenceofRPstate} to get a state $\tilde{\phi}$, which is exact and reflection positive. By Theorem \ref{thm:canonicalpurificationquasiequivalence}, $\tilde{\phi}$ is unitarily equivalent to $\psi$. There exists a strictly $C$-local $D_N$-invariant automorphism that transforms $\tilde{\phi}$ into a trivial state. Therefore, $\tilde{\phi}$ represents a trivial element in $\InvPhasesZN$.

By Lemma \ref{lma:existenceofSRPZNinvObservable}, there exists a strictly reflection positive $D_N$-invariant observable $a$. Therefore, by Lemma \ref{lma:perturbationisstrictlyRP}, the state $\omega$ defined by $ \omega (x) = \tilde{\phi}(a^* x a)/\tilde{\phi}(a^* a)$ is strictly reflection positive and represents a trivial element in $\InvPhasesZN$. 

Let $\Omega$, $\Psi$ be vectors representing the states $\omega$, $\psi$ in the GNS Hilbert space associated with $\omega$. Since $\omega$ and $\psi$ are strictly reflection positive, we have $\lal \Omega, \Psi\ral > 0$. Thus, $\psi$ also represents a trivial element in $\InvPhasesZN$, because otherwise $\lal \Omega, \Psi \ral = 0$.

\end{proof}

\section{A refined index for 2d invertible phases}   \label{sec:index}

In \cite{sopenko2024}, a $U(1)$-valued index $\omega_N$ for invertible phases has been defined, where $N$ is an arbitrary positive integer. It was explained that this index is essentially given by an SPT index\footnote{An SPT index for 2d invertible states for an arbitrary on-site finite group action has been defined in \cite{sopenko2021,ogata2021h}. But it was not pointed out that this index can be used to define an invariant of invertible phases in the absence of any symmetry.} for an $S_N$ symmetric group action on a stack of $N$ copies of a given invertible state. The latter interpretation leads to the quantization condition $(\omega_N)^{12} = 1$. Based on the analogy with the anomalies of conformal field theories where a similar invariant appears \cite{johnson2019moonshine}, it was conjectured that this index provides a microscopic definition of the chiral central charge mod $24$. It was also shown in \cite{sopenko2024}, that for fermionic systems one can define an analogous index that allows us to distinguish phases represented by quasi-free invertible states with Chern number $\nu \bmod 48 \neq 0$. That gives an additional evidence for the relation between $\omega_N$ and $c_- \bmod 24$. 

The construction of \cite{sopenko2024}, however, was designed with an intent to define a finer invariant than an SPT index. As was argued earlier for the special class of Laughlin's states \cite{gromov2016geometric}, the braiding properties of twist defects (which are also called genons in \cite{barkeshli2013twist}) in an $N$-layered system are sensitive to the chiral central charge $c_-$. Since twist defects can be defined microscopically for any invertible quantum many-body system and since their statistics characterized by "topological spin" $\theta_N$ is robust against arbitrary local deformations, one could hope to formalize $\theta_N$ as an invariant of the phase that is defined without an assumption of conformal symmetry. Since it detects $c_-$ when conformal symmetry emerges, it should provide a microscopic alternative for $c_-$.

It was shown in \cite{sopenko2024}, that $\omega_N = (\theta_N)^N$ indeed provides an invariant of the phase. It was also pointed out (see Remark 3.2 and Remark 1.2 in \cite{sopenko2024}) that $\theta_N$ {\it would be} an invariant of the phase if it was not for the ambiguity in the choice of a $\ZZ/N$-charge for a twist defect state that changes $\theta_N$ by an $N$-th root of unity. This ambiguity is closely related to an ambiguity in the angular momentum of an invertible state invariant under $\ZZ/N$-rotational symmetry (see Section \ref{ssec:rotationalinvariance}). 

With the canonical choice of the angular momentum from Section \ref{sec:AngularMomentum} provided by reflection-positivity, we are now in a position to define a refined index $\theta_N \in U(1)$ that satisfies $(\theta_N)^N = \omega_N$. In this section, we basically repeat the construction of $\omega_N$ from \cite{sopenko2024} with an adjustment that allows to define a refined invariant $\theta_N$.

\begin{remark} \label{rmk:thetaNandcminus}
When conformal symmetry emerges on the boundary of a material, twist defect states in the bulk correspond to sectors of twist field operators in the conformal field theory describing the edge modes. There is a well-known relation between the chiral central charge $c_-$ and the conformal spin $\theta^{\text{CFT}}_{N}$ of twist field operators \cite{knizhnik1987analytic}
$$
\theta^{\text{CFT}}_{N} = e^{2 \pi i \frac{c_-}{24} \l N - \frac{1}{N} \r}.
$$
Moreover, for holomorphic conformal field theories which correspond to invertible states in the bulk, we have a quantization condition $c_- \in 8 \ZZ$ that implies $(\theta^{\text{CFT}}_N)^{3N} = 1$. 

It would be interesting to determine whether the same dependence on $N$ and the same quantization condition holds for the index $\theta_N$ we define in this paper. If true, $\theta_N$ provides a microscopic definition of $c_-$ for invertible phases.
\end{remark}

\begin{remark}
We believe that the same idea can be used to define a $U(1)$-valued invariant of a more general class of 2d topological phases that exhibit anyons, once a proper mathematical framework characterizing such phases is developed. The main physical ingredients that make this invariant well-defined are the existence of twist defects and their mobility, which allows us to define the exchange phase factor. 

Under the hypothesis that such phases are described by a modular tensor category $\mathcal{C}$ and a chiral central charge $c_-$, the expected characterization of twist defects and their statistics can be deduced from the data of an $S_N$-crossed braided tensor category $\mathcal{C}^{\boxtimes N}$ together with a "twist" given by an element in $H^3(S_N,U(1))$ (see \cite{barkeshli2019symmetry}). We emphasize that the information about the twist is not contained in $\mathcal{C}$ but is an additional piece of data fixed by $c_-$. This is consistent with the fact that $\mathcal{C}$ can detect only $c_- \bmod 8$.
\end{remark}

\subsection{The index}

\begin{figure}
\centering
\begin{tikzpicture}[scale=.5]
\filldraw[color=orange!10, fill=orange!10, ultra thick] (0,0) -- (2,-4) -- (-2,-4) -- cycle;
\filldraw[color=orange!10, fill=orange!10, ultra thick] (0,0) -- (3.4641-1,2+3.4641/2) -- (3.4641+1,2-3.4641/2) -- cycle;
\filldraw[color=orange!10, fill=orange!10, ultra thick] (0,0) -- (-3.4641+1,2+3.4641/2) -- (-3.4641-1,2-3.4641/2) -- cycle;
\draw[orange, very thick] (0,0) -- (3.4641+1,2-3.4641/2);
\draw[orange, very thick] (0,0) -- (3.4641-1,2+3.4641/2);
\draw[orange, very thick] (0,0) -- (-3.4641-1,2-3.4641/2);
\draw[orange, very thick] (0,0) -- (-3.4641+1,2+3.4641/2);
\draw[orange, very thick] (0,0) -- (-2,-4);
\draw[orange, very thick] (0,0) -- (2,-4);

\draw[gray, very thick] (0,0) -- (3.4641,2);
\draw[gray, very thick] (0,0) -- (-3.4641,2);
\draw[gray, very thick] (0,0) -- (0,-4);

\node  at (0,-1.5*4) {$B_{12}$}; 
\node  at (1.5*1.7321*4/2,1.5*4/2) {$B_{01}$};
\node  at (-1.5*1.7321*4/2,1.5*4/2) {$B_{20}$};
\node  at (0,1.5*3) {$A_0$}; 
\node  at (-1.5*1.7321*3/2,-1.5*3/2) {$A_2$};
\node  at (1.5*1.7321*3/2,-1.5*3/2) {$A_1$};

\filldraw [gray] (0,0) circle (5pt);

\end{tikzpicture}
\caption{The cones $A_0, A_1, A_2$ are separated by gray lines. The orange lines correspond to the boundaries of the cones $B_{01}, B_{12}, B_{20}$.}
\label{fig:cones}
\end{figure}
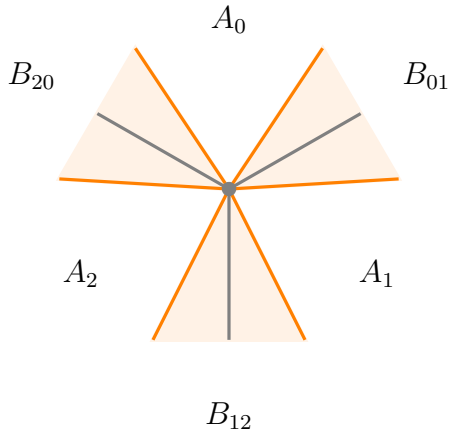

We consider the same setup as in Section \ref{ssec:twistdefect}. Let $\omega$ be an invertible state on $\RR^2$, and let $N$ be a positive integer. Let $\cstar{A}$ be a stack of $N$ copies of the system and let $\psi = \omega^{\otimes N}$ be the corresponding state. The symmetric group $S_N$ acts on $\cstar{A}$ via on-site automorphisms in a natural way and preserves the state $\psi$. Let $\rho$ be the generator of $\ZZ/N \subset S_N$ subgroup that corresponds to the cyclic permutation $\{1,2,...,N\} \to \{N,1,2,...,N-1\}$.

In Section \ref{sec:AngularMomentum}, we explained how to associate a $\ZZ/N$-rotationally invariant state with any twist defect state. We say that a twist defect state $\varphi$ has a {\it canonical $\ZZ/N$-charge} if the corresponding $\ZZ/N$-rotationally invariant state has a canonical angular momentum. 

Let $A$ be a cone having a half-line $l_+$ as one of its sides, and let $\varphi$ be a twist defect state associated with $l_+$. Then if $\varphi$ has a canonical $\ZZ/N$-charge, then so is the twist defects state $\varphi \rho_{A}$ (which is associated with another side of $A$), because these two states give the same $\ZZ/N$-rotationally invariant state up to isometry of $\RR^2$.

Let us choose a good conical cover $\{A_0, A_1, A_2, B_{01}, B_{12}, B_{20}\}$ as depicted on Fig. \ref{fig:cones} with the apex and the boundaries of the cones not intersecting the locations of the sites of the lattice. We define $\rho_a := \rho_{A_a}$ which is the restriction of the automorphism $\rho$ to the cone $A_a$. We let $l_{20}$, $l_{01}$, $l_{12}$ be the half-lines between the pairs of cones $A_0$, $A_1$, $A_2$. Possibly after trivial extension of $\omega$, we can choose $\rho$-invariant automorphisms $\kappa_{20}$, $\kappa_{01}$, $\kappa_{12}$ on $B_{20}$, $B_{01}$, $B_{12}$, respectively, such that twist defect states $\psi \kappa_{a(a+1)}^{-1}$ associated with $l_{a(a+1)}$ have a canonical $\ZZ/N$-charge and the automorphisms $w_a = \kappa_{(a-1)a}^{-1} \rho_a \kappa_{a(a+1)}$ preserve the unitary equivalence class of $\psi$. Note that $w_a$ commute with each other.

Let $P$, $W_a$ be the unitary operators implementing $\rho$, $w_a$, respectively, in the GNS representation associated with $\psi$ so that $P^N = 1$. Because the twist defect states $\psi \kappa_{a(a+1)}^{-1}$ have a canonical $\ZZ/N$-charge, we have $P W_a P^{-1} = W_a$. Because $w_a$ commute with each other, we have
$$
W_0 W_1 = \theta_N W_1 W_0
$$
for some $\theta_N \in U(1)$.

\begin{prop} \label{prop:ThetaIndex}
For a given invertible phase, $\theta_N \in U(1)$ constructed above does not depend on the choice of a representative $\omega$, a good conical cover $\{A_0, A_1, A_2, B_{01}, B_{12}, B_{20}\}$, automorphisms $\kappa_{a(a+1)}$, and unitary operators $W_a$. It defines a homomorphism $\theta_N: \InvPhases \to U(1)$.
\end{prop}
\begin{lemma} \label{lma:supercommutingimplementers}
Let $\omega$ be an invertible state. Let $B_L$, $B_R$ be disjoint cones with the same apex, and let $\beta_L$, $\beta_R$ be automorphisms on $B_L$, $B_R$, respectively, which preserve the unitary equivalence class of $\omega$. Let $U_L$, $U_R$ be the unitary operators implementing $\beta_L$, $\beta_R$ in the GNS representation of $\omega$. Then $U_L U_R = U_R U_L$.
\end{lemma}

\begin{proof}
Since we can always stack $\omega$ with its inverse, without loss of generality we can assume that $\omega$ is in a trivial phase. 

Let $C$ be a sufficiently fine good conical cover and let $\alpha$ be a strictly $C$-local automorphism, such that $\omega \alpha$ is unitarily equivalent to a trivial state. Then $\alpha^{-1} \beta_L \alpha$, $\alpha^{-1} \beta_R \alpha$ are automorphisms on slightly bigger disjoint cones which preserve the unitary equivalence class of a trivial state. The unitary operators implementing these automorphisms commute in the GNS representation of a trivial state, that implies the statement of the lemma.


\end{proof}
\begin{proof}[Proof of Proposition \ref{prop:ThetaIndex}]
The operators $W_a$ only have phase factor ambiguity which clearly does not affect $\theta_N$. 

Suppose $\tilde{\kappa}_{a(a+1)}$ is a different choice for $\kappa_{a(a+1)}$. Then $(\kappa^{-1}_{a(a+1)} \tilde{\kappa}_{a(a+1)})$ preserves the unitary equivalence class of $\psi$, and therefore can be implemented by a unitary operator $K_{20}$ in the GNS representation. Since both states $\psi \kappa^{-1}_{a(a+1)}$, $\psi \tilde{\kappa}^{-1}_{a(a+1)}$ have a canonical $\ZZ/N$-charge, we have $P K_{a(a+1)} P^{-1} = K_{a(a+1)}$. By Lemma \ref{lma:supercommutingimplementers}, $K_{a(a+1)}$ commutes with $W_{a+2}$. For a new choice of $W_a$, we have $\tilde{W}_a = K_{(a-1)a}^{-1} W_a K_{a(a+1)}$. Then
$$
\tilde{W}_1 \tilde{W}_0 = K_{01}^{-1} W_1 K_{12} K_{20}^{-1} W_0 K_{01} = K_{20}^{-1} K_{01}^{-1} W_1  W_0 K_{01} K_{12} = K_{20}^{-1}  W_1  W_0 K_{12} = \theta_{N}^{-1} \tilde{W}_0 \tilde{W}_1,
$$
where we have used $W_0 W_1 K_{01} W_0^{-1} W_1^{-1} = P K_{01} P^{-1} = K_{01}$. Thus, $\theta_N$ is independent of the choice of $\kappa_{a(a+1)}$.

To show independence of the good conical cover, we note that we can replace the cones $A_a$, $B_{a(a+1)}$ in the construction by more general regions $\tilde{A}_a = A_a \setminus Y_a$, $\tilde{B}_{a(a+1)} = B_a \setminus Y_{a(a+1)}$ for some bounded regions $Y_a$, $Y_{a(a+1)}$ which we can call "asymptotic cones". A different choice of a single asymptotic cone $\tilde{A}_a$ or $\tilde{B}_{a(a+1)}$ (without changing the other asymptotic cones) can be equivalently described as a different choice of $\kappa_{a(a+1)}$, and therefore, gives the same $\theta_N$. By changing asymptotic cones step by step we can relate any two good conical covers, therefore $\theta_N$ is unambiguous.

Let us choose $C$ with the same apex as the cones $A_a$, $B_{a(a+1)}$ and with the cones having sufficiently small bases (compared to the bases of $A_a$, $B_{a(a+1)}$). Suppose $\omega'$ is another invertible state in the same phase as $\omega$. Trivial extension of the states does not affect the computation of $\theta_N$. We can assume without loss of generality that $\omega$ and $\omega'$ are related by a $C$-local automorphism. Let $\alpha$ be the corresponding $C$-local automorphism relating $\psi$ and $\psi' = (\omega')^{\otimes N}$. Then, to compute $\theta_N$ for $\psi'$, we can use automorphisms $w'_a = \alpha^{-1} w_a \alpha$, which produce the same commutations relations for the unitary operators implementing them. Thus, $\theta_N$ is an invariant of invertible phases. It is also manifest from the construction that the assignment $\theta_N: \InvPhases \to U(1)$ is a homomorphism.
\end{proof}

It is clear from the construction of $\omega_N$ from \cite{sopenko2024}, that $(\theta_N)^N = \omega_N$. It was shown that $\omega_N$ is related to the SPT index for $S_N$-symmetry taking values in $H^3(S_{N \geq 6},U(1)) = \ZZ/3 \times \ZZ/4 \times \ZZ/2 \times \ZZ/2$ that leads to the quantization condition $(\omega_N)^{12} = 1$. Therefore, we have 
$$
(\theta_N)^{12 N} = 1.    
$$
Note that this quantization condition is weaker than the one expected from conformal field theory $(\theta_N)^{3N}=1$ (see Remark \ref{rmk:thetaNandcminus}). We leave it as an open question whether the stronger quantization condition holds.

Suppose $\omega_N = 1$. In addition to the canonical angular momentum for a $\ZZ/N$-rotationally invariant state, there is an angular momentum for which twist defect states have a trivial topological spin. The invariant $\theta_N$ is given by the relative $\ZZ/N$-charge of the corresponding twist defects. Let $N' = m N$. A $\ZZ/N'$-rotationally invariant state with the canonical angular momentum would also have a canonical angular momentum for $\ZZ/N$-rotational symmetry. Similarly, a state that corresponds to a twist defect state for $\ZZ/N'$-symmetry with a trivial topological spin would also correspond to such a state with $\ZZ/N$-symmetry. It follows that we have the following relation $\theta_N = (\theta_{mN})^{m}.$ Sequences $\{\theta_N\}_{N \in \NN}$, $\theta_1 = 1$ satisfying this relation are in one-to-one correspondence with pro-finite integers $\hat{\ZZ} := \text{Hom}(\mathbb{Q}/\ZZ, U(1))$ which contain integers $\ZZ \subset \hat{\ZZ}$ via $n \to e^{2 \pi i n (\,\cdot \,)}$, $n \in \ZZ$. Thus, we can define an index $\mu \in \hat{\ZZ}$ for invertible phases by taking twelve copies of a given invertible state so that $\omega_N=1$. Conjecturally, $\mu \in 4 \ZZ \subset \hat{\ZZ}$.

\appendix

\section{Reflection positivity and canonical purification} \label{app:canonicalpurification}

Let $\cstar{B} = M_n(\CCC)$ be a matrix $C^*$-algebra. For $x \in \cstar{B}$, we denote the complex conjugate matrix by $\bar{x}$. Let $\cstar{A} = \cstar{B} \otimes \cstar{B}$. For any state $\psi$ on $\cstar{B}$ described by a density matrix $\rho$ (i.e. $\psi(x) = \Tr (\rho x)$), there are many pure states $\omega$ on $\cstar{A}$, satisfying $\omega(x \otimes 1) = \psi(x)$. Such states are called {\it purifications} of $\psi$. If in addition we require that $\omega(x \otimes \bar{x}) \geq 0$ for any $x \in \cstar{B}$, then such $\omega$ is unique and is called a {\it canonical purification}. It is canonical because we have fixed an anti-linear automorphism $a \to \bar{a}$ of $\cstar{B}$. If $\rho = \sum_{i=1}^{n} \lambda_i \xi_i \otimes \xi^*_i$, $\lambda_i \geq 0$, $\xi_i \in \CCC^n$, $\lal \xi_i, \xi_j \ral = \delta_{i,j}$\footnote{In this paper, we denote the inner product in a Hilbert space $\hilb{H}$ by $\lal\,\cdot\,,\,\cdot\,\ral$. Our convention is that the inner product is anti-linear in the first argument: $\lal b y, a x \ral = a \bar{b} \lal y,x \ral$, $a,b \in \CCC$, $x,y \in \hilb{H}$.}, then the vector in $\CCC^d \otimes \CCC^d$ representing the canonical purification is given by $\sum_{i=1}^{n} \sqrt{\lambda_i} \xi_i \otimes \bar{\xi}_i \in \CCC^n \otimes \CCC^n$. Such vectors form a positive self-dual cone in $\CCC^n \otimes \CCC^n$, and we have a bijection between vectors in this cone and states on $\cstar{B}$.

This construction has a generalization to infinite-dimensional unital $C^*$-algebras as explained in \cite{woronowicz1972purification, woronowicz1973purification}. We review this construction in this appendix. We then analyze the situation when we have an additional symmetry of $\cstar{A}$ that corresponds to $\ZZ/N$ rotations of a quantum system in various physical applications. As a byproduct, we obtain some statements about 1d quantum spin chains satisfying reflection positivity (see Section \ref{sapp:1dspinchain}).

\subsection{Canonical purification}

Let $\cstar{B}$ be a unital $C^*$-algebra. The opposite $C^*$-algebra $\cstar{B}^{op}$ has the same space of elements and $*$-operation as $\cstar{B}$ but with the multiplication being given by $a \circ b := ba$. We denote the element in $\cstar{B}^{op}$ corresponding to $a^* \in \cstar{B}$ by $\bar{a}$. Note that $\bar{a} \circ \bar{b} = \overline{ab}$.

Let $\cstar{A} = \cstar{B} \otimes \cstar{B}^{op}$, and let $j$ be the anti-linear automorphism of $\cstar{A}$ defined by $j(a \otimes \bar{b}) = b \otimes \bar{a}$. A state $\omega$ on $\cstar{A}$ is called {\it $j$-positive} if for any $x \in \cstar{B}$ we have $\omega(x \otimes \bar{x}) \geq 0$. 
\begin{remark}
If $\omega$ is $j$-positive, then it is $j$-invariant, i.e. $\omega(x) = \overline{\omega(j(x))}$. Indeed, since $\omega((\bar{a}+\bar{b})\otimes (a + b)) \geq 0$ and $\omega((\bar{a} - i \bar{b})\otimes (a + i b)) \geq 0$, the real part of $-\omega(\bar{b} \otimes a) + \omega(\bar{a} \otimes b)$ and the imaginary part of $\omega(\bar{b} \otimes a) + \omega(\bar{a} \otimes b)$ vanish. Therefore, $\omega(\bar{b} \otimes a) = \overline{\omega(\bar{a} \otimes b)}$.
\end{remark}
A state $\omega$ on $\cstar{A}$ is called {\it exact}\footnote{We follow the terminology of \cite{woronowicz1972purification}. The same condition is often called Haag duality.} if the von Neumann algebras $\cstar{R} = \pi_{\omega}(\cstar{B})''$, $\cstar{R}^{op} = \pi_{\omega}(\cstar{B}^{op})''$ generated by $\cstar{B}$, $\cstar{B}^{op}$ in the GNS representation $(\pi_{\omega},\hilb{H}_{\omega})$ corresponding to $\omega$ are commutants of each other: $\cstar{R}' = \cstar{R}^{op}$. A pure state $\omega$ on $\cstar{A}$ is called a {\it purification} of a state $\psi$ on $\cstar{B}$ if for any $x \in \cstar{B}$, we have $\psi(x) = \omega(x \otimes \bar{1})$. A state $\psi$ on $\cstar{B}$ is called {\it factorial} if the von Neumann algebra $\pi_{\psi}(\cstar{B})''$ in the GNS representation $(\pi_{\psi},\hilb{H}_{\psi})$ corresponding to $\psi$ has a trivial center.

\begin{theorem}[Theorem 1.1 \cite{woronowicz1972purification} together with Theorem 1.1 \cite{woronowicz1973purification}] \label{thm:canonicalpurification}
Let $\psi$ be a factorial state on $\cstar{B}$. Then there exists a unique purification $\omega$ of $\psi$ which is exact and $j$-positive. 
\end{theorem}
We call the purification $\omega$ from Theorem \ref{thm:canonicalpurification} the {\it canonical purification} of $\psi$. The usefulness of this notion is clear from
\begin{theorem}[Theorem 1.2 \cite{woronowicz1972purification}] \label{thm:canonicalpurificationquasiequivalence}
Two factorial states $\psi_1, \psi_2$ on a $C^*$-algebra $\cstar{B}$ are quasi-equivalent if and only if their canonical purifications $\omega_1, \omega_2$ are unitarily equivalent. 
\end{theorem}

Suppose $\omega$ is an exact $j$-positive state on $\cstar{A}$ with the corresponding GNS representation $(\pi_{\omega},\hilb{H}_{\omega}, \Omega)$. We will say that it is {\it strictly $j$-positive} if for any non-zero $x \in \cstar{R}$ we have $\lal \Omega, x J x J \Omega \ral > 0$, where $J$ is the anti-unitary operator implementing $j$ in $\hilb{H}_{\omega}$, i.e. $\pi_{\omega}(j(a)) = J \pi_{\omega}(a) J$ for any $a \in \cstar{A}$ and $J \Omega = \Omega$. If $\omega$ is pure, it is equivalent to requiring that $\Omega$ is a cyclic separating vector for $\cstar{R}$ as evident from Tomita-Takesaki theory. In this case, reflection positive states correspond to vectors in the natural positive cone $\hilb{P}$ that is given by the closure of the set $\{x J x J \Omega\}_{x \in \cstar{R}}$. For any $\Psi \in \hilb{P}$ representing a reflection positive state, we have $\lal \Psi,\Omega \ral > 0$.

\subsection{Krein-Rutman theorem}

Let $\cstar{R} \subset B(\hilb{H})$ be a von Neumann algebra acting on a separable Hilbert space that admits a cyclic separating vector $\Omega \in \hilb{H}$. We let $J$ be the modular conjugation. We say that a positive trace-class operator $\rho$ on $\hilb{H}$ is reflection positive if $\Tr (\rho x J x J) \geq 0$ for any $x \in \cstar{R}$. We say that it is strictly reflection positive if $\Tr (\rho x J x J) > 0$ for any non-zero $x \in \cstar{R}$.

\begin{prop} \label{prop:perronfrobenius}
Let $\rho \in B(\hilb{H})$ be a density matrix operator (i.e., a positive trace-class operator with unit trace) which is strictly reflection positive. Then the spectral radius $r(\rho)$ of $\rho$ is an eigenvalue of $\rho$ with an eigenvector $\xi$ that is cyclic and separating for $\cstar{R}$. In particular, the projection to $\xi$ is strictly reflection positive.
\end{prop}
\begin{proof}
The proof is an application of the weak version of the Krein-Rutman theorem (e.g., see Theorem 19.2 in \cite{deimling2013nonlinear}).

Let $\hilb{H}_{\text{sa}} \subset \hilb{H}$ be the real Hilbert space of vectors $v\in\hilb{H}$ satisfying $J v = v$. We denote the natural positive cone for $\cstar{R}$ by $\mathcal{P}$. The cone $\mathcal{P}$ is total $\hilb{H}_{\text{sa}} = \mathcal{P}-\mathcal{P}$. 

Let $\bar{\hilb{H}}$ be the complex conjugate Hilbert space and let $\bar{\cstar{R}}$ be the corresponding to $\cstar{R}$ von Neumann algebra acting on it. For a Schmidt decomposition $\rho = \sum_{i} \lambda_i \xi_i \otimes \xi_i^*$, we have $\Tr(\rho a) = \lal \Psi, (a \otimes 1) \Psi\ral$, $a \in B(\hilb{H})$ where $\Psi = \sum_{i} \sqrt{\lambda_i} \xi_i \otimes \bar{\xi}_i \in \hilb{H} \otimes \bar{\hilb{H}}$. The vector $\Psi$ is cyclic and separating for the algebra generated by $\cstar{R}$ and $\bar{\cstar{R}}$. Hence, $\lal v \otimes \bar{w}, \Psi \ral = \lal v, \rho w \ral > 0$ for any $v,w \in \hilb{P}$. It follows that $\rho \hilb{P} \subset \hilb{P}$ and therefore, by Krein-Rutman theorem, $\rho$ has an eigenvector $v$ with the eigenvalue $r(\rho)$. Since for any $w \in \hilb{P}$ we have $\lal w, v \ral = \lal w, \rho v \ral/r(\rho) >0$, the vector $v$ is cyclic and separating for $\cstar{R}$.
\end{proof}

\subsection{Reflection positive observables}

In the following, we only consider the case when $\cstar{B} = \bigotimes_{i=1}^{\infty} M_{n_i}(\CCC)$, $n_i \in \NN$ is a uniformly hyperfinite algebra. In this case, we can identify $\cstar{B}^{op}$ with $\cstar{B}$ and encode this identification into the data of the anti-linear automorphism $j$ of $\cstar{A} = \cstar{B} \otimes \cstar{B}$. In particular, we can let $j$ act via $j(a \otimes b) = \bar{b} \otimes \bar{a}$, where $\bar{a}, \bar{b} \in \cstar{B}$ are complex conjugate observables. We also often say {\it (strictly) reflection positive} instead of (strictly) $j$-positive when the anti-linear automorphism $j$ is clear from the context. We denote the inner product on $\cstar{A}$ induced by the unique tracial state $\tau$ on $\cstar{A}$ by $\lal a,b \ral_{\tau} := \tau(a^* b)$.

Following \cite{jaffe2017reflection}, we introduce the cone $\cstar{K}_+$ which is the norm closure of the positive linear span of observables $\{x \otimes \bar{x}\}_{x \in \cstar{B}}$. Elements of $\cstar{K}_+$ are called {\it reflection positive observables}. The element $a \in \cstar{K}_+$ is called {\it strictly reflection positive} if $\tau(a^* b) > 0$ for any non-zero $b \in \cstar{K}_+$. Note that if $a,b \in \cstar{K}_+$, then $a^*,b^*,ab \in \cstar{K}_+$. Also note that if $a \in \cstar{K}_+$ is strictly reflection positive, then so are $a^*$ and $ab$ for any $b \in \cstar{K}_+$.

Suppose $a \in \cstar{K}_+$ is strictly reflection positive. Since $\tau(a^* (x \otimes \bar{x})))>0$ for $x \in \cstar{B}$, there exists $\lambda > 0$, such that $a - \lambda (x \otimes \bar{x})$ is reflection positive. In particular, for some $\lambda > 0$, $(a-\lambda)$ is reflection positive, that implies $\omega(a)>0$ for any reflection positive state $\omega$ on $\cstar{A}$.

\begin{lemma} \label{lma:perturbationisstrictlyRP}
Let $\omega$ be a pure exact reflection positive state on $\cstar{A}$, and let $a \in \cstar{K}_+$ be a strictly reflection positive observable. Then the state $\omega_a$ defined by $\omega_a(b) := \omega(a^* b a)/\omega(a^* a)$, $b \in \cstar{A}$ is strictly reflection positive.
\end{lemma}
\begin{proof}
The observable $a^* a$ is strictly reflection positive. Hence, $\omega(a^* a) >0$  and the state $\omega_a$ is well-defined.

Let us choose a GNS representation $(\pi_{\omega}, \hilb{H}_{\omega},\Omega)$ for $\omega$. The state $\omega_a$ is a vector state in this representation, and therefore is pure. Since $a^* b a$ is reflection positive for any reflection positive $b$, $\omega_a$ is reflection positive. Since $\omega_a$ is unitarily equivalent to an exact state, it is also exact.

Suppose there exists $y \in \cstar{R}$, $\cstar{R}=\pi_{\omega}(\cstar{B})''$ such that $\lal \Omega, \pi_{\omega}(a^*) y J y J \pi_{\omega}(a) \Omega \ral = 0$. Since for any $b = x \otimes \bar{x}$, $x \in \cstar{B}$ we can find $\lambda>0$, such that $a-\lambda b$ is reflection positive, we have $\lal \Omega, \pi_{\omega}(b^*) y J y J \pi_{\omega}(b) \Omega \ral = 0$. Since the linear span of vectors $\{\pi_{\omega}(x \otimes \bar{x}) \Omega\}_{x \in \cstar{B}}$ is dense in $\hilb{H}_{\omega}$, we have $y=0$.

\end{proof}

\subsection{Rotationally invariant reflection positive states}

Let $\cstar{B} = \bigotimes_{i=1}^{\infty} M_{n_i}$, $n_i \in \NN$ be a uniformly hyperfinite algebra, and let $\cstar{A} = \cstar{B}^{\otimes 2N}$, $N \in \NN$. We let $\ups$ be the "shift by two sites" automorphism $\ups(a_1 \otimes a_2 \otimes ... \otimes a_{2N}) = a_{2N-1} \otimes a_{2N} \otimes a_1 \otimes ... \otimes a_{2N-2}$ of $\cstar{A}$ and $j$ be a reflection anti-linear automorphism $j(a_1 \otimes a_2 \otimes ... \otimes a_{2N}) = \bar{a}_{2N} \otimes \bar{a}_{2N-1} \otimes ... \otimes \bar{a}_{1}$.

Let us show that strictly positive $\ups$-invariant observables exist.
\begin{lemma} \label{lma:IsingChain}
Let $\cstar{B} \cong M_{2}(\CCC)$ and let $X_i, Y_i, Z_i \in \cstar{A}$ be Pauli operators acting on the $i$-th factor in $\cstar{B}^{\otimes 2N} = \cstar{A}$. We set $X_{N+1} := X_1$, $Z_{N+1} := Z_1$. The projector to the ground state vector (i.e., the eigenvector with the lowest eigenvalue) of a 1d transverse field Ising model Hamiltonian $H = - \sum_{i=1}^{2N} X_i X_{i+1} - \sum_{i=1}^{2N} Z_i$ defines a strictly reflection positive element of $\cstar{A}$ .
\end{lemma}
\begin{proof}
The fact that the state is reflection positive follows from Proposition \ref{prop:jaffeRP} and Proposition \ref{prop:PerronFrob1dspinchains}. To show strict reflection positivity, it is enough to show the faithfulness of the restriction of the state to a half-chain $\cstar{B}^{\otimes N}$.

We can use the Jordan-Wigner transformation to relate the ground state of $H$ to the ground state $\psi$ of the fermionic system with the Hamiltonian $H_{Majorana} = i \sum_{i=1}^{4 N-1} c_i c_{i+1} - i c_{4 N} c_{1}$, where $c_{2j-1} = Z_1 ... Z_{j-1} X_j$, $c_{2j} = Z_1 ... Z_{j-1} Y_j$, $j = 1,...,2N$ are Majorana operators generating the corresponding CAR algebra with relations $\{c_j, c_k\} = 2\delta_{j,k}$, $j,k = 1,...,4N$. Since the Hamiltonian is quadratic in fermionic operators $\{c_j\}$, the ground state is quasi-free and we can explicitly compute $\psi(c_j c_k) = \delta_{j,k} + B_{j,k}$, where
$$
B_{j,k} = i \frac{1-(-1)^{k-j}}{4N \sin \l \pi \frac{k-j}{4 N} \r}
$$
for $j \neq k$ and $B_{j,j} = 0$. Let $A_{j,k} = -i B_{j,4N+1-k}$, $j,k=1,...,2N$, and let $f:\{1,...,2N\} \to \CCC$ be an arbitrary non-zero function. We have
\begin{multline*}
\sum_{j=1}^{2N} \sum_{k=1}^{2N} \bar{f_j} A_{j,k} f_k = \sum_{j=1}^{2N} \sum_{k=1}^{2N} \bar{f_j} f_k \frac{1+(-1)^{j+k}}{4 N \sin\l \frac{\theta_j+\theta_k}{2} \r} = \sum_{j=1}^{2N} \sum_{k=1}^{2N} \bar{g_j} g_k \frac{1+(-1)^{j+k}}{4 N (t_j + t_k)} = \\ = \frac{1}{4N} \int_{0}^{\infty} d u \sum_{j=1}^{2N} \sum_{k=1}^{2N} (1+(-1)^{j+k}) \overline{g_j e^{-u t_j}} g_k e^{-u t_k} = \\ = \frac{1}{4N} \int_{0}^{\infty} d u \l \left| \sum_{l=1}^{N} g_{2l-1} e^{-u t_{2l-1}} \right|^2 + \left| \sum_{l=1}^{N} g_{2l} e^{-u t_{2l}} \right|^2 \r > 0,
\end{multline*}
where $\theta_j = \pi(j-1/2)/2N$, $t_j = \tan(\theta_j/2)$, $g_j = f_j \sqrt{1+t_j^2}$. It follows that the matrix $A$ is positive definite. Therefore, $\psi$ is strictly reflection positive and the restriction of the state $\psi$ to the half-chain is faithful.
\end{proof}

\begin{lemma} \label{lma:existenceofSRPZNinvObservable}
There exists a strictly reflection positive $\ups$-invariant observable in $\cstar{A}$.
\end{lemma}
\begin{proof}
Suppose $\cstar{B} = M_{d}(\CCC)$. By Lemma \ref{lma:IsingChain}, there is a strictly reflection positive $\ups$-invariant observable when $d=2^k$ for some $k \in \NN$. If $d<2^k$, let $p$ be the projection in $M_{2^k}(\CCC)$ to $\CCC^d \subset \CCC^{2^k}$. Then projecting a strictly reflection positive $\ups$-invariant state on $M_{2^k}(\CCC)^{\otimes 2N}$ using $p^{\otimes 2N}$, we get a strictly reflection positive $\ups$-invariant positive linear functional on $\cstar{A}$. After normalization, we get a strictly reflection positive $\ups$-invariant state, and the projector to this state defines a strictly reflection positive $\ups$-invariant observable. 

Suppose now $\cstar{B} = \bigotimes_{i=1}^{\infty} M_{n_i}(\CCC)$, $n_i \in \NN$ is a uniformly hyperfinite $C^*$-algebra. Let us choose a sequence of strictly reflection positive $\ups$-invariant observables $\{b_i\}_{i \in \NN}$, $\|b_i\|=1$ for $M_{n_i}(\CCC)^{\otimes 2N}$. Then an observable $a$ that is a limit of $\{ a_i = \prod_{m=1}^i (1 + b_m/m^2) \}$ is strictly reflection positive and $\ups$-invariant.
\end{proof}

Combined with Lemma \ref{lma:perturbationisstrictlyRP}, we get the following

\begin{prop} \label{prop:makingstrictlyRP}
Let $\psi$ be an $\ups$-invariant exact reflection positive  state on $\cstar{A}$. Then there exists a unitarily equivalent strictly reflection positive $\ups$-invariant state on $\cstar{A}$.
\end{prop}

\subsection{Applications for 1d quantum spin systems}

\label{sapp:1dspinchain}
Let us mention some elementary applications of the proven statements for finite one-dimensional quantum spin systems. 

Let $\hilb{H} = (\CCC^d)^{\otimes 2N}$ be the Hilbert space of a 1d quantum spin chain of length $2N$ with on-site Hilbert spaces $\CCC^{d}$. The Hilbert space of the first (the last) $N$ sites is denoted $\hilb{H}_{L}$ ($\hilb{H}_R$), so that $\hilb{H} = \hilb{H}_L \otimes \hilb{H}_R$. The algebras $\cstar{A}$, $\cstar{A}_L$, $\cstar{A}_R$ are the algebras of operators on $\hilb{H}$, $\hilb{H}_L$, $\hilb{H}_R$, respectively. We let $J$ be the anti-unitary operator on $\hilb{H}$ defined by $J(v_1 \otimes v_2 \otimes ... \otimes v_{2N}) = \bar{v}_{2N} \otimes \bar{v}_{2N-1} \otimes ... \otimes \bar{v}_{1}$. It corresponds to a $CRT$-symmetry of the system. 

A state described by a density matrix $\rho$ is reflection positive if $\Tr(\rho x J x J) \geq 0$ for any $x \in \cstar{A}_L$. It is strictly reflection positive if the inequality is strict $\Tr(\rho x J x J) > 0$ for any non-zero $x \in \cstar{A}_L$. A vector $\Omega \in \hilb{H}$ is (strictly) reflection positive if the projector to $\Omega$ is (strictly) reflection positive and $J \Omega = \Omega$.

In \cite{jaffe2017reflection}, the following characterization of reflection positive density matrices has been given
\begin{prop}[Theorem IV.10 in \cite{jaffe2017reflection}] \label{prop:jaffeRP}
Let $H \in \cstar{A}$ be a self-adjoint observable that corresponds to the Hamiltonian of the system. Suppose $H = H_L + H_0 + H_R$ for some $H_{L} \in \cstar{A}_L$, $H_{R} = J H_L J \in \cstar{A}_R$, and a reflection positive $(-H_0) \in \cstar{A}$. Let $\rho = Z^{-1} e^{-\beta H}$, $Z = \Tr e^{-\beta H}$, $\beta \in (0,\infty)$ be the density matrix of the Gibbs state at inverse temperature $\beta$. Then $\rho$ is reflection positive.   
\end{prop}

Let us show that reflection positivity of the Gibbs states implies reflection positivity of the ground state.

\begin{prop} \label{prop:PerronFrob1dspinchains}
Let $H \in \cstar{A}$ be a self-adjoint observable that corresponds to the Hamiltonian of the system. Let $\rho = Z^{-1} e^{-\beta H}$, $Z = \Tr e^{-\beta H}$, $\beta \in (0,\infty)$ be the density matrix of the Gibbs state at inverse temperature $\beta$. If $\rho$ is reflection positive, then $H$ has a ground state that is also reflection positive. If $\rho$ is strictly reflection positive, then $H$ has a unique ground state that is also strictly reflection positive. Moreover, there are no other reflection positive eigenstates except for the ground state.
\end{prop}

\begin{proof}
The proof is an application of the weak (for a reflection positive $\rho$) and the strong (for a strictly reflection positive $\rho$) version of the Krein-Rutman theorem (e.g. see Theorem 19.2 and Theorem 19.3 in \cite{deimling2013nonlinear}).

Let $\hilb{H}_{\text{sa}} \subset \hilb{H}$ be the real subspace of vectors $J v = v$, and let $\Omega$ be a strictly positive vector. We denote the positive cone of the convex hull of vectors $\{x J x J \Omega\}_{x \in \cstar{A}_L}$ by $\hilb{P}$. It is a total self-dual cone with non-empty interior. 

If $\rho$ is reflection positive, then for any reflection positive observable $x$, the operator $\rho x$ is reflection positive. Hence, $\rho \hilb{P} \subseteq \hilb{P}$. If $\rho$ is strictly reflection positive, then the operator $\rho x$ is strictly reflection positive, and $\rho \hilb{P}$ belongs to the interior of $\hilb{P}$. Therefore, in both cases $\rho$ satisfies the criteria of the Krein-Rutman theorems that imply the statements of the proposition.
\end{proof}

Suppose now we have a periodic 1d quantum spin chain of length $2N$. We let $U$ be the "shift by two sites" unitary operator on $\hilb{H}$ defined by $U(v_1 \otimes v_2 \otimes ... \otimes v_{2N}) = v_{2N-1} \otimes v_{2N} \otimes v_1 \otimes ... \otimes v_{2N-2}$. The operator $U$ satisfies $U^N = 1$ and generates $\ZZ/N$-rotational symmetry.

\begin{prop} \label{prop:AngularMomenumofRP1dspinchain}
Any reflection positive $\ZZ/N$-rotationally invariant pure state on $\cstar{A}$ represented by a vector $\Omega \in \hilb{H}$ has vanishing angular momentum, i.e. $U \Omega = \Omega$.
\end{prop}
\begin{proof}
Suppose $U \Omega = \lambda \Omega$ for some $\lambda\in \CCC$, $\lambda^N = 1$. We can choose a reflection positive product vector $\Omega_0 = v^{\otimes 2N}$, $v \in \CCC^{d}$ satisfying $U \Omega_0 = \Omega_0$. By Proposition \ref{prop:makingstrictlyRP}, there exists a pure strictly reflection positive rotationally invariant state on $\cstar{A}$. Let $\Psi \in \hilb{H}$ be the corresponding strictly reflection positive vector which satisfies $U \Psi = \mu \Psi$ for some $\mu \in \CCC$, $\mu^N =1$. We have $\lal\Psi, \Omega \ral >0 $ and $\lal\Psi, \Omega_0 \ral >0 $. Therefore, $\lambda = \mu = 1$.
\end{proof}

\printbibliography

@article{powers1967representations,
  title={Representations of uniformly hyperfinite algebras and their associated von Neumann rings},
  author={Powers, Robert T},
  journal={Annals of Mathematics},
  pages={138--171},
  year={1967},
  publisher={JSTOR}
}

@article{woronowicz1972purification,
  title={On the purification of factor states},
  author={Woronowicz, Stanislaw L},
  journal={Communications in Mathematical Physics},
  volume={28},
  number={3},
  pages={221--235},
  year={1972},
  publisher={Springer}
}

@article{woronowicz1973purification,
  title={On the purification map},
  author={Woronowicz, SL},
  journal={Communications in Mathematical Physics},
  volume={30},
  number={1},
  pages={55--67},
  year={1973},
  publisher={Springer}
}

@article{kitaev2006anyons,
  title={Anyons in an exactly solved model and beyond},
  author={Kitaev, Alexei},
  journal={Annals of Physics},
  volume={321},
  number={1},
  pages={2--111},
  year={2006},
  publisher={Elsevier}
}

@article{Matsui,
  title={Boundedness of entanglement entropy and split property of quantum spin chains},
  author={Matsui, Taku},
  journal={Rev. Math. Phys.},
  volume={25},
  number={09},
  pages={1350017},
  year={2013},
  publisher={World Scientific}
}

@article{bachmann2012automorphic,
  title={Automorphic equivalence within gapped phases of quantum lattice systems},
  author={Bachmann, Sven and Michalakis, Spyridon and Nachtergaele, Bruno and Sims, Robert},
  journal={Communications in Mathematical Physics},
  volume={309},
  number={3},
  pages={835--871},
  year={2012},
  publisher={Springer}
}

@article {ogata2019classification,
    AUTHOR = {Ogata, Yoshiko},
     TITLE = {A classification of pure states on quantum spin chains
              satisfying the split property with on-site finite group
              symmetries},
   JOURNAL = {Trans. Amer. Math. Soc. Ser. B},
  FJOURNAL = {Transactions of the American Mathematical Society. Series B},
    VOLUME = {8},
      YEAR = {2021},
     PAGES = {39--65},
}

@article{kapustin2014symmetry,
  title={Symmetry protected topological phases, anomalies, and cobordisms: beyond group cohomology},
  author={Kapustin, Anton},
  journal={arXiv preprint arXiv:1403.1467},
  year={2014}
}

@article{freed2016reflection,
  title={Reflection positivity and invertible topological phases},
  author={Freed, Daniel S and Hopkins, Michael J},
  journal={arXiv preprint arXiv:1604.06527},
  year={2016}
}

@ARTICLE{ogata2021h,
       author = {{Ogata}, Yoshiko},
        title = "{A $H^{3}(G,{\mathbb T})$-valued index of symmetry protected topological phases with on-site finite group symmetry for two-dimensional quantum spin systems}",
         year = 2021,
          doi = {10.48550/arXiv.2101.00426},
archivePrefix = {arXiv},
       eprint = {2101.00426},
 primaryClass = {math-ph},
}

@ARTICLE{kapustin2020classification,
       author = {{Kapustin}, Anton and {Sopenko}, Nikita and {Yang}, Bowen},
        title = "{A classification of invertible phases of bosonic quantum lattice systems in one dimension}",
      journal = {Journal of Mathematical Physics},
         year = 2021,
       volume = {62},
       number = {8},
        pages = {081901},
archivePrefix = {arXiv},
       eprint = {2012.15491},
 primaryClass = {quant-ph},
}

@ARTICLE{sopenko2021,
       author = {{Sopenko}, Nikita},
        title = "{An index for two-dimensional SPT states}",
      journal = {J. Math. Phys.},
         year = 2021,
       volume = {62},
       number = {11},
        pages = {111901},
}

@article{johnson2019moonshine,
  title={The moonshine anomaly},
  author={Johnson-Freyd, Theo},
  journal={Communications in Mathematical Physics},
  volume={365},
  number={3},
  pages={943--970},
  year={2019},
  publisher={Springer}
}

@article{knizhnik1987analytic,
  title={Analytic fields on Riemann surfaces. II},
  author={Knizhnik, VG},
  journal={Communications in mathematical physics},
  volume={112},
  number={4},
  pages={567--590},
  year={1987},
  publisher={Springer}
}

@article{naaijkens2022split,
  title={The split and approximate split property in 2D systems: stability and absence of superselection sectors},
  author={Naaijkens, Pieter and Ogata, Yoshiko},
  journal={Communications in Mathematical Physics},
  volume={392},
  number={3},
  pages={921--950},
  year={2022},
  publisher={Springer}
}

@article{gromov2016geometric,
  title={Geometric defects in quantum Hall states},
  author={Gromov, Andrey},
  journal={Physical Review B},
  volume={94},
  number={8},
  pages={085116},
  year={2016},
  publisher={APS}
}

@article{barkeshli2013twist,
  title={Twist defects and projective non-Abelian braiding statistics},
  author={Barkeshli, Maissam and Jian, Chao-Ming and Qi, Xiao-Liang},
  journal={Physical Review B—Condensed Matter and Materials Physics},
  volume={87},
  number={4},
  pages={045130},
  year={2013},
  publisher={APS}
}

@article{jaffe2017reflection,
  title={Reflection positive doubles},
  author={Jaffe, Arthur and Janssens, Bas},
  journal={Journal of Functional Analysis},
  volume={272},
  number={8},
  pages={3506--3557},
  year={2017},
  publisher={Elsevier}
}

@article{sopenko2024,
  title={An index for invertible phases of two-dimensional quantum spin systems},
  author={Sopenko, Nikita},
  journal={arXiv preprint arXiv:2410.02059},
  year={2024}
}

@article{osterwalder1973axioms,
  title={Axioms for Euclidean Green's functions},
  author={Osterwalder, Konrad and Schrader, Robert},
  year={1973}
}

@book{deimling2013nonlinear,
  title={Nonlinear functional analysis},
  author={Deimling, Klaus},
  year={2013},
  publisher={Springer Science \& Business Media}
}

@misc{KitaevBravyiChain,
  author = {Alexei Kitaev},
  note = {Private communication},
%  title = {A short, descriptive title if available (optional)}
}

@misc{BravyiCounterexample,
  author = {Sergey Bravyi},
  note = {Unpublished}
}

@article{kim2022chiral,
  title={Chiral central charge from a single bulk wave function},
  author={Kim, Isaac H and Shi, Bowen and Kato, Kohtaro and Albert, Victor V},
  journal={Physical Review Letters},
  volume={128},
  number={17},
  pages={176402},
  year={2022},
  publisher={APS}
}

@article{barkeshli2019symmetry,
  title={Symmetry fractionalization, defects, and gauging of topological phases},
  author={Barkeshli, Maissam and Bonderson, Parsa and Cheng, Meng and Wang, Zhenghan},
  journal={Physical Review B},
  volume={100},
  number={11},
  pages={115147},
  year={2019},
  publisher={APS}
}

\end{document}